\documentclass{llncs}

\usepackage{amsmath,amssymb,amsbsy,stmaryrd}
\usepackage[ruled,linesnumbered,noline,noend]{algorithm2e}
\usepackage{tikz}
\usetikzlibrary{shapes,arrows}

\newtheorem{dfn}{Definition}%[section]
\newtheorem{thm}{Theorem}%[section]
\newtheorem{lem}{Lemma}%[section]

\newcommand{\todo}[1]{\marginpar{\scriptsize\textcolor{red}{#1}}}
\newcommand{\rem}[1]{{\scriptsize{{\textcolor{magenta}{[#1]}}}}}
\renewcommand{\todo}[1]{}
\renewcommand{\rem}[1]{}
\renewcommand{\note}[1]{}
\newcommand{\marek}[1]{{{\textcolor{blue}{#1}}}}
\renewcommand{\marek}[1]{}

\newcommand{\false}{\ensuremath{\mathit{false}}}
\newcommand{\true}{\ensuremath{\mathit{true}}}
\newcommand{\ite}{\ensuremath{\mathbf{ite}}}

\newcommand{\var}[1]{{\tt #1}}
\newcommand{\sym}[1]{\ensuremath{\underline{#1}}}
\newcommand{\prm}[1]{\ensuremath{\llbracket #1 \rrbracket}}
\newcommand{\ese}[1]{\ensuremath{\langle #1 \rangle}}
\newcommand{\bld}[1]{\ensuremath{\boldsymbol{#1}}}
\newcommand{\compose}{\ensuremath{\diamond}}

\newcommand{\KwInSep}{ \\ \hspace{1cm} }

\newcommand{\lmark}[1]{\nlset{#1~~~~}}

\DontPrintSemicolon

\newcommand{\aset}{\ensuremath{\longleftarrow}}
\newcommand{\sat}{\texttt{satisfiable}}

\begin{document}

%{{{ Title + Authors + Abstract

\frontmatter

\title{Compact Symbolic Execution\thanks{This is a full version of the paper
    accepted to ATVA 2013.}}
\author{Jiri Slaby \and Jan Strej\v{c}ek \and Marek Trt\'{\i}k}

%\institute{Faculty of Informatics, Masaryk University\\
%  Botanick\'{a} 68a, 60200 Brno, Czech Republic\\
%  \email{\{slaby,strejcek,trtik\}@fi.muni.cz}
%}
\institute{Masaryk University, Brno, Czech Republic\\
  \email{\{slaby,strejcek,trtik\}@fi.muni.cz}
}

\maketitle

\begin{abstract}
  We present a generalisation of King's symbolic execution technique called
  \emph{compact symbolic execution}. It proceeds in two steps.  First, we
  analyse cyclic paths in the control flow graph of a given program,
  independently from the rest of the program. Our goal is to compute a so
  called \emph{template} for each such a cyclic path. A template is a
  declarative parametric description of all possible program states, which
  may leave the analysed cyclic path after any number of iterations along
  it. In the second step, we execute the program symbolically with the
  templates in hand. The result is a \emph{compact} symbolic execution
  tree. A compact tree always carry the same information in all its leaves
  as the corresponding classic symbolic execution tree. Nevertheless, a
  compact tree is typically substantially smaller than the corresponding
  classic tree. There are even programs for which compact symbolic execution
  trees are finite while classic symbolic execution trees are infinite.
%  We present a generalisation of King's symbolic execution technique called
%  \emph{compact symbolic execution}. It is based on a concept of templates:
%  a template is a declarative parametric description of a cycle in the
%  control flow graph of a program. 
%  % program part, whose any execution has some regularities in program
%  % states along it. Typical sources of such regularities are program loops
%  % and recursive calls.
%  Using the templates, we are able to build a symbolic execution tree where
%  a single node represents arbitrarily many iterations of the corresponding
%  cycle without loss of any information. This leads to a considerable
%  reduction of size of the symbolic execution tree. There are even programs
%  for which the resulting compact symbolic execution trees are finite while
%  the classic symbolic execution trees are infinite.
\end{abstract}

% \category{D.2.4}{Software Engineering}{Software/Program Verification}
% \category{\\D.2.5}{Software Engineering}{Testing and Debugging}
% \category{\\F.3.1}{Logics and Meanings of Programs}{Specifying and
% Verifying and Reasoning about Programs}

% \terms
% Reliability, Verification, Algorithms

% \keywords
% Symbolic execution, Path Explosion Problem

%}}}

%{{{ Introduction

\section{Introduction} \label{sec:Introduction}

Symbolic execution~\cite{Kin76,How77} is a program analysis method
originally suggested for enhanced testing. While testing runs a program on
selected input values, symbolic execution runs the program on symbols that
represent arbitrary input values. As a result, symbolic execution explores
all execution paths. On one hand-side, this means that symbolic execution
does not miss any error. On the other hand-side, symbolic execution applied
to real programs hardly ever finishes as programs typically have a huge (or
even infinite) number of execution paths. This weakness of symbolic
execution is known as \emph{path explosion problem}. The second weakness of
symbolic execution comes from the fact that it calls SMT solvers to decide
which program paths are feasible and which are not. The SMT queries are
often formulae of theories that are hard to decide or even
undecidable. Despite the two weaknesses, there are several successful
bug-finding tools based on symbolic execution, for example
\textsc{Klee}~\cite{CDE08}, \textsc{Exe}~\cite{Cad08},
\textsc{Pex}~\cite{TdH08}, or \textsc{Sage}~\cite{GLM08}.

This paper introduces the \emph{compact symbolic execution} that partly
solves the path explosion problem.  We build on the observation that one of
the main sources of the problem are program cycles.  Indeed, many execution
paths differ just in numbers of iterations along program cycles. Hence,
before we start symbolic execution, we detect cyclic paths in the control
flow graph of a given program and we try to find a \emph{template} for each
such a cyclic path.  A template is a declarative parametric description
(with a single parameter $\kappa$) of all possible program states produced
by $\kappa\ge 0$ iterations along the cyclic path followed by any execution
step leading outside the cyclic path. The target program locations of such
execution steps are called \emph{exits} of the cyclic path.
% to any so-called exit program location. 
%
% A template is a declarative parametric description (with a single parameter
% $\kappa$) of all possible program states at the program locations leaving
% the analysed cyclic path after any $\kappa\ge0$ iterations along~it.

The compact symbolic execution proceeds %builds a symbolic execution tree
just like the classic symbolic execution until we enter a cyclic path for
which we have a template. Instead of executing the cyclic path, we can apply
the template to jump directly to exits of the cyclic path. At each exit, we
obtain a program state with a parameter $\kappa$. This parametric program
state represents all program states reached by execution paths composed of a
particular path to the cycle, $\kappa$ iterations along the cycle, and the
execution step leading to the exit. Symbolic execution then continues from
these program states in the classic way again.

Hence, compact symbolic execution reduces the path explosion problem as it
explores at once all execution paths that differ only in numbers of
iterations along the cyclic paths for which we have templates. As we will
see later, a price for this reduction comes in deepening the other weakness
of symbolic execution: while SMT queries of standard symbolic execution are
always quantifier-free, each application of a template adds one universal
quantifier to the SMT queries of compact symbolic execution. Although SMT
solvers fail to decide quantified queries significantly more often than
queries without quantifiers, our experimental results show that this
trade-off is acceptable as compact symbolic execution is able to detect more
errors in programs than the classic one. Moreover, future advances in SMT
solving can make the disadvantage of compact symbolic execution even
smaller.

% Honzo, je mi jasne, ze tento text budes hledat. Nez ho ale 
% opet odkomentujes, tak skutecne zvaz, jestli je tento prehled
% pro 15 strankovy clanek se 6 sekcema, z nichz polovina (Intro,
% Related Work a Conclusion) jsou zcela typicky ve vsech clancich,
% OPRAVDU nutne potrebny.
% The rest of the paper is organized as follows.  The next section presents
% the basic idea of the compact symbolic execution using a simple example.
% Section~\ref{sec:technique} provides a formal description of the technique.
% Section~\ref{sec:experiments} is devoted to experimental implementation and
% experimental results.  The last two sections discuss the related work and
% conclude the paper.

%}}}
%{{{ Basic idea

\section{Basic Idea}\label{sec:basic}

This section presents basic ideas of compact symbolic execution. To
illustrate the ideas, we use a simple program represented by the flowgraph
of Figure~\ref{fig:linSrch}(a). The program implements a standard linear
search algorithm. It returns the least index \texttt{i} in the array
\texttt{A} such that \texttt{A[i]=x}. If \texttt{x} is not in \texttt{A} at
all, then the result is \texttt{-1}. In both cases the result is saved in
the variable \var{r}.
Before we describe the compact symbolic execution, we briefly recall the
classic symbolic execution~\cite{Kin76}.

\begin{figure}[!htb]
  \begin{tabular}{ccc}
    \hspace{-4ex}\centering
\tikzstyle{start} = [regular polygon,regular polygon sides=3,
    regular polygon rotate=180,thick,draw,minimum size=9mm,inner sep=0pt]
\tikzstyle{target} = [regular polygon,regular polygon sides=3,
    regular polygon rotate=0,thick,draw,minimum size=9mm,inner sep=0pt]
\tikzstyle{loc} = [circle,thick,draw,minimum size=6mm]
\tikzstyle{pre} = [<-,shorten <=1pt,>=stealth',semithick]
\tikzstyle{post} = [->,shorten <=1pt,>=stealth',semithick]
\footnotesize
\begin{tikzpicture}[node distance=1.0cm]
    \pgfsetmovetofirstplotpoint
    % \pgfplothandlerclosedcurve
    % %\pgfplothandlerlineto
    % \pgfplotstreamstart
    %     \pgfplotstreampoint{\pgfpoint{-0.1cm}{-0.5cm}}
    %     \pgfplotstreampoint{\pgfpoint{-2.4cm}{-3.1cm}}
    %     \pgfplotstreampoint{\pgfpoint{0.15cm}{-3.5cm}}
    %     \pgfplotstreampoint{\pgfpoint{1.4cm}{-2.0cm}}
    % \pgfplotstreamend
    % \pgfsetfillcolor{lightgray!65}
    % \pgfusepath{fill}
    % %\pgfusepath{stroke}

    \pgfplothandlerclosedcurve
    %\pgfplothandlerlineto
    \pgfplotstreamstart
        \pgfplotstreampoint{\pgfpoint{0.1cm}{-1.0cm}}
        \pgfplotstreampoint{\pgfpoint{-1.4cm}{-2.5cm}}
        \pgfplotstreampoint{\pgfpoint{0.1cm}{-4.0cm}}
        \pgfplotstreampoint{\pgfpoint{0.55cm}{-2.5cm}}
    \pgfplotstreamend
    \pgfsetfillcolor{gray!30}
    \pgfusepath{fill}
    \pgfusepath{stroke}

    \pgfsetfillcolor{black}
    
    \node [start] (a) {$a$};
    \node [loc] (b) [inner sep=2.5pt,below of=a,yshift=-5mm] {$b$}
        edge [pre] node [label=right:\texttt{i:=0}] {} (a);
    \node [loc] (c) [left of=b,yshift=-10mm] {$c$}
        edge [pre, bend left=20]
                node [label=above:\texttt{i<n~~~~}] {} (b);
    \node [loc] (f) [right of=b,inner sep=1.9pt,yshift=-10mm] {$f$}
        edge [pre, bend right=20]
                node [label=above:\texttt{~~~~i>=n}] {} (b);
    \node [loc] (d) [right of=c,yshift=-10mm] {$d$}
        edge [pre, bend left=20]
            node [label=left:\texttt{A[i]!=x}] {} (c)
        edge [post, bend right=20]
                node [label=left:\texttt{++i}] {} (b);
    \node [loc] (e) [left of=c,xshift=-5mm,yshift=-8mm] {$e$}
        edge [pre, bend left=20] node
            [label=above:\texttt{A[i]=x~~~~~}] {} (c);
    \node [target] (g) [below of=d,yshift=-5mm] {$g$}
        edge [pre, bend left=20]
                node [label=left:\texttt{r:=i~~}] {} (e)
        edge [pre, bend right=20]
                node [label=below:\texttt{~~~~~~r:=-1}] {} (f);
    %\node [] (blank) [below of=g,yshift=-1.73cm] {};
    \node [] (lab) [below of=g,yshift=-5mm] {(a)};
\end{tikzpicture}  & 
    \hspace{-4.5ex}\centering
\tikzstyle{start} = [regular polygon,regular polygon sides=3,
    regular polygon rotate=180,thick,draw,inner sep=0.5pt]
\tikzstyle{target} = [regular polygon,regular polygon sides=3,
    regular polygon rotate=0,thick,draw,inner sep=0.2pt]
\tikzstyle{loc} = [circle,thick,draw,inner sep=1.0pt,minimum size=5mm]
\tikzstyle{pre} = [<-,shorten <=1pt,>=stealth',semithick]
\tikzstyle{post} = [->,shorten <=1pt,>=stealth',semithick]
\footnotesize
\begin{tikzpicture}[node distance=0.9cm]
%     \draw[color=black,thick,dotted] (-1.6,-4.8)--(-1.6,-5.2);
%     \draw[color=black,thick,dotted] (-0.2,-5)--(-0.2,-5.4);
%     \draw[color=black,thick,dotted] (0.8,-4.3)--(0.8,-4.7);

    \pgfsetmovetofirstplotpoint

    \pgfplothandlerclosedcurve
    \pgfplotstreamstart
      \pgfplotstreampoint{\pgfpoint{0.4cm}{-1cm}}
      \pgfplotstreampoint{\pgfpoint{0.4cm}{-6.87cm}}
      \pgfplotstreampoint{\pgfpoint{-0.4cm}{-6.87cm}}
      \pgfplotstreampoint{\pgfpoint{-0.4cm}{-1cm}}
    \pgfplotstreamend
    \pgfsetfillcolor{gray!30}
    \pgfusepath{fill}

    \pgfsetfillcolor{black}

    \node [loc,inner sep=1.5pt] (a) {$a$};

    \node [loc] (b) [inner sep=1.3pt,below of=a] {$b$}
        edge [pre] node {} (a);
    \node [loc] (c) [below of=b] {$c$}
        edge [pre] node [label=left:{\scriptsize $0<\sym{n}$}] {} (b);
    \node [loc,inner sep=0.5pt] (f) [right of=c,xshift=9mm] {$f$}
        edge [pre] node [label=above:{\scriptsize $~~0\geq\sym{n}$}] {} (b);
    \node [loc] (g) [below of=f] {$g$}
        edge [pre] node {} (f);
    \node [loc] (d) [below of=c] {$d$}
        edge [pre] node [label=right:{\scriptsize
            $~\sym{A}(0)\neq\sym{x}$}] {} (c);
    \node [loc] (e) [left of=d,xshift=-2mm] {$e$}
        edge [pre] node [label=above:{\scriptsize
            $\sym{A}(0)=\sym{x}~~~~~~~~~~~~$}] {} (c);
    \node [loc] (h) [below of=e] {$g$}
        edge [pre] node {} (e);

    \node [loc] (b1) [below of=d] {$b$}
        edge [pre] node {} (d);
    \node [loc] (c1) [below of=b1] {$c$}
        edge [pre] node [label=left:{\scriptsize $1<\sym{n}$}] {} (b1);
    \node [loc,inner sep=0.5pt] (f1) [right of=c1,xshift=9mm]{$f$}
        edge [pre] node [label=above:{\scriptsize $~~1\geq\sym{n}$}] {} (b1);
    \node [loc] (g1) [below of=f1] {$g$}
        edge [pre] node {} (f1);
    \node [loc] (d1) [below of=c1] {$d$}
        edge [pre] node [label=right:{\scriptsize
            $~\sym{A}(1)\neq\sym{x}$}] {} (c1);
    \node [loc] (e1) [left of=d1,xshift=-2mm] {$e$}
        edge [pre] node [label=above:{\scriptsize
            $\sym{A}(1)=\sym{x}~~~~~~~~~~~~$}] {} (c1);
    \node [loc] (h1) [below of=e1] {$g$}
        edge [pre] node {} (e1);

    \node [loc] (b2) [below of=d1] {$b$}
        edge [pre] node {} (d1);

    \node [circle,thick,inner sep=0pt,minimum size=5mm] (b3) [below of=b2] {}
        edge [pre] node {} (b2);
    \node [circle,thick,inner sep=0pt] (b33) [below of=b2,yshift=1.5mm] {\scriptsize $\vdots$};

    \node [] (lab) [below of=b3,yshift=2mm] {(b)};
\end{tikzpicture}  &
    \hspace{-1.5ex}\input{fig_linSrchCSET}
  \end{tabular}
  \caption{(a) A flowgraph \texttt{linSrch(A,n,x)}. (b) Classic symbolic
    execution tree of \texttt{linSrch}. (c) Compact symbolic execution tree
    of \texttt{linSrch}.}
  \label{fig:linSrch}
\end{figure}

\subsubsection{Classic Symbolic Execution} 
Symbolic execution runs a program over \emph{symbols} representing arbitrary
input values. For each input variable $\var{v}$, we denote a symbol passed
to it as $\sym{v}$. A \emph{program state} is a triple $(l,\theta,\varphi)$
consisting of a current program location $l$ in the flowgraph, a
\emph{symbolic memory} $\theta$, and a \emph{path condition}
$\varphi$. $\theta$ assigns to each program variable its current symbolic
value, i.e.~an expression over the symbols. For example, if the first
instruction of a program is the assignment \texttt{i:=2*n+x}, then
$\theta(\var{i})=2\sym{n}+\sym{x}$ after its execution. The path condition
$\varphi$ is a quantifier-free first order logic formula representing a
necessary and sufficient condition on symbols to drive the execution along
the currently executed path. $\varphi$ is initially $\true$ and it can be
updated at program branchings. For example, in a location with two out-edges
labelled by \texttt{x>n+5} and \texttt{x<=n+5}, we instantiate the
conditions with use of the current $\theta$ and we check whether the current
path condition $\varphi$ implies their validity. Namely, we ask for validity
of implications $\varphi \rightarrow \theta(\var{x}) > \theta(\var{n}) + 5$
and $\varphi \rightarrow \theta(\var{x}) \leq \theta(\var{n}) + 5$.
%\[
%\varphi\implies(\theta(\var{x})>\theta(\var{n})+5)
%~~~~~~~\textrm{and}~~~~~~~
%\varphi\implies(\theta(\var{x})\le\theta(\var{n})+5).
%\] 
If the first implication is valid, the symbolic execution continues along
the first branch. If the second implication is valid, the symbolic execution
continues along the second branch. If none of them is valid, it means that
we can follow either of the two branches.  Hence, the symbolic execution
forks in order to execute both branches.  In this case, we update the path
condition on the first branch to $\varphi~\wedge~\theta(\var{x}) >
\theta(\var{n}) + 5$ and the one on the second branch to $\varphi~\wedge~
\theta(\var{x}) \leq \theta(\var{n}) + 5$. Note that the whole program state
is forked into two states in this case.

Due to the forks, symbolic execution is traditionally represented by a tree
called \emph{classic symbolic execution tree}. Nodes of the tree are
labelled by program states computed during the execution.
% \todo{For brevity of notation, we identify nodes with program states
% labelling them.}
Edges of the tree correspond to transitions between program states labelling
their end nodes.  In Figure~\ref{fig:linSrch}(b), there is a classic
symbolic execution tree of the flowgraph from Figure~\ref{fig:linSrch}(a).
For readability of symbolic execution tree figures, nodes are marked only
with current program locations instead of full program states.  In addition,
we label branching edges by instances of the corresponding branching
conditions in the flowgraph. These labels allow us to reconstruct the path
condition for each node in the tree: it is the conjunction of labels of all
edges along the path from the root to the node. Note that contents of
symbolic memories are not depicted in the figure.
%\nb{je tahle veta nutna, kdyz pred tim piseme, ze
%tam nedavame program states for readability?}

\subsubsection{Overall Effect of Cyclic Paths}
% Compact symbolic execution employs templates computed for cyclic paths of
% an analyzed flowgraph.
If we look at the flowgraph of Figure~\ref{fig:linSrch}(a), we immediately
see that locations $b,c,d$ and edges between them form a cyclic path
highlighted by a grey region. All executions entering the path (at location
$b$) proceed in the same way: each execution performs $\kappa$ iterations
along the cyclic path (for some $\kappa\ge0$) and continues either along the
edge $(b,f)$ or along the edges $(b,c)$ and $(c,e)$ to leave it.
Compact symbolic execution aims to effectively exploit the uniformity of all
executions along this cyclic path. To do so, we need to find a unified
declarative description of the \emph{effect} of all executions along the
cyclic path on a symbolic memory and a path condition.  We analyse the
cyclic path, together with all the edges allowing to leave it, separately
from the rest of the flowgraph. First we introduce symbols for all variables
in the isolated part of the program, since they all are now input variables
to the part. In our example, we introduce symbols
$\sym{n},\sym{x},\sym{i},\sym{A}$ representing the values of the
corresponding variables $\var{n},\var{x},\var{i},\var{A}$ at the entry
location $b$, before the first iteration. We emphasise that the introduced
symbols do \emph{not} represent inputs to the whole flowgraph, but rather
symbolic values of the corresponding variables at the moment of entering the
cyclic path at the location $b$ via the edge $(a,b)$.

Now we study the effect of $\kappa$ iterations along the cyclic path. One can
see that each iteration increases the value of $\var{i}$ by one while values
of the other variables keep unchanged. Hence, after $\kappa$ iterations, the
value of $\var{i}$ is $\sym{i}+\kappa$.  Formally, the effect of $\kappa$
iterations of the cycle on values of all variables is described by the
following \emph{parametric symbolic memory} $\theta_*\prm{\kappa}$ with the
parameter $\kappa$:
\[
\setlength{\arraycolsep}{2pt}
\begin{array}{rclp{3ex}rcl}
  \theta_*\prm{\kappa}(\var{n})&=&\sym{n},&&
  \theta_*\prm{\kappa}(\var{x})&=&\sym{x},\\
  \theta_*\prm{\kappa}(\var{i})&=&\sym{i}+\kappa,&&
  \theta_*\prm{\kappa}(\var{A})&=&\sym{A}.
\end{array}
\]

Further, we formulate a \emph{parametric path condition}
$\varphi_*\prm{\kappa}$ representing the path condition after $\kappa$
iterations along the cyclic path. To perform all these $\kappa$ iterations
along the cyclic path, both conditions \texttt{i<n} and \texttt{A[i]!=x}
along the path have to be valid in each of $\kappa$ iterations. Therefore,
the path condition after $\kappa$ iterations has the form
\begin{align*}
  \sym{i}<\sym{n}&~~\wedge~~\sym{A}(\sym{i})\not=\sym{x}~~\wedge\\
  \wedge~~\sym{i}+1<\sym{n}&~~\wedge~~\sym{A}(\sym{i}+1)\not=\sym{x}~~\wedge\\
  &~~~~\vdots\\
  \wedge~~\sym{i}+(\kappa-1)<\sym{n}&~~\wedge~~\sym{A}(\sym{i}+(\kappa-1))\not=\sym{x},
%  \textrm{,}
\end{align*}
where $\tau$-th line, $\tau\in\{0,1,\ldots,\kappa-1\}$, consists of two
predicates which are instances of the conditions \texttt{i<n} and
\texttt{A[i]!=x} respectively after $\tau$ iterations of the cyclic path,
i.e.~during the $(\tau+1)$-st iteration. Unfortunately, the conjunction
above is not a first order formula as its length depends on the parameter
$\kappa$, whose value can be arbitrary. The conjunction can be equivalently
expressed by the following universally quantified formula:
% To uniformly describe all such path
% conditions for all values of $\kappa$, we have to use a quantified formula:
\[
  \forall\tau(0\le\tau<\kappa\rightarrow
  (\sym{i}+\tau<\sym{n}\,\wedge\,\sym{A}(\sym{i}+\tau)\neq\sym{x})).
\]
If we now add to the formula above the obvious fact that we cannot
iterate the cyclic path negative number of times (i.e.~$\kappa \geq
0$), we get the resulting parametric path condition
$\varphi_*\prm{\kappa}$ as
\[
  \varphi_*\prm{\kappa}~=~\kappa\ge0~\wedge~\forall
  \tau(0\le\tau<\kappa\rightarrow
  (\sym{i}+\tau<\sym{n}~\wedge~\sym{A}(\sym{i}+\tau)\neq\sym{x})).
\]

Finally, we use $\theta_*\prm{\kappa}$ and $\varphi_*\prm{\kappa}$ to define
symbolic memory $\theta_{b\!f}\prm{\kappa}$ and path condition
$\varphi_{b\!f}\prm{\kappa}$ describing the effect of $\kappa$ iterations of
the cyclic path followed by leaving it through the edge $(b,f)$, and
similarly $\theta_{ce}\prm{\kappa},\varphi_{ce}\prm{\kappa}$ with the
analogous information for leaving the cyclic path through the edge
$(c,e)$. As the edges $(b,f),(b,c),(c,e)$ do not modify any variable, we
immediately get
$\theta_{b\!f}\prm{\kappa}=\theta_{ce}\prm{\kappa}=\theta_*\prm{\kappa}$. Further,
$\varphi_{b\!f}\prm{\kappa}$ and $\varphi_{ce}\prm{\kappa}$ are conjunctions
of $\varphi_*\prm{\kappa}$ with the instances of the conditions on the edge
$(b,f)$ or on the edges $(b,c),(c,e)$, respectively. Hence, the path
conditions $\varphi_{b\!f}\prm{\kappa},\varphi_{ce}\prm{\kappa}$ are defined
as follows:
\[
\setlength{\arraycolsep}{0pt}
\begin{array}{rcl}
  \varphi_{b\!f}\prm{\kappa}&~=~&\varphi_*\prm{\kappa}~\wedge~\sym{i}+\kappa\ge\sym{n}\\
  \varphi_{ce}\prm{\kappa}&=&\varphi_*\prm{\kappa}~\wedge~\sym{i}+\kappa<\sym{n}
    ~\wedge~\sym{A}(\sym{i}+\kappa)=\sym{x}
\end{array}
\]

The overall effect of the considered cyclic path with its exit edges
is now fully described by a so-called \emph{template} consisting of
the entry location $b$ to the cyclic path and two triples
$(f,\theta_{b\!f}\prm{\kappa},\varphi_{b\!f}\prm{\kappa})$ and
$(e,\theta_{ce}\prm{\kappa},\varphi_{ce}\prm{\kappa})$, one for each
exit edge from the cyclic path.
%The first elements of these triples are expressions $(cdb)^\kappa f$
%and $(cdb)^\kappa ce$ representing sequences of locations passed in
%correspoding cycle executions. Technically, we only need to remember the
%correspoding exit locations $f$ and $e$ to know where to continue the
%symbolic execution from. We include the whole expressions only for clarity
%of presentation (and we will leave this practice in the formal definition).
Note that the triples have the same structure and meaning as program
states in classic symbolic execution. The only difference is that the
triples are parametrised by the parameter $\kappa$.

\subsubsection{Compact Symbolic Execution}
The template is used during \emph{compact symbolic execution} of the
program. The execution starts at the location $a$ of the flowgraph.  The
\emph{compact symbolic execution tree} initially consists of a single
node %, $v$ say,
labelled by the initial state $(a,\theta_I,\true)$, where $\theta_I$ is the
initial symbolic memory assigning to each input variable \var{v} the
corresponding symbol $\sym{v}$. Now we execute the instruction \texttt{i:=0}
of the flowgraph edge $(a,b)$ using the classic symbolic execution. The tree
is extended with a single successor node, say $u$, labelled with a program
state $(b,\theta',\varphi')$. As we have a template for the location $b$, we
can instantiate it instead of executing the original program. The node $u$
gets one successor for each triple of the template. The triple
$(f,\theta_{b\!f}\prm{\kappa},\varphi_{b\!f}\prm{\kappa})$ generates a
successor node labelled by a program state
$(f,\theta'_{b\!f}\prm{\kappa},\varphi'_{b\!f}\prm{\kappa})$. Note that we
cannot use $(f,\theta_{b\!f}\prm{\kappa},\varphi_{b\!f}\prm{\kappa})$
directly as $\theta_{b\!f}\prm{\kappa},\varphi_{b\!f}\prm{\kappa}$ describe
executions starting just at the entry location $b$, while
$\theta'_{b\!f}\prm{\kappa},\varphi'_{b\!f}\prm{\kappa}$ have to reflect the
effect of the executions starting at $a$. We create
$\theta'_{b\!f}\prm{\kappa},\varphi'_{b\!f}\prm{\kappa}$ by composing
$\theta_{b\!f}\prm{\kappa},\varphi_{b\!f}\prm{\kappa}$ with
$\theta',\varphi'$. The composition is precisely described in the following
section. In our simple program, $\theta',\varphi'$ reflect only the effect
of assignment \texttt{i:=0}. Thus, $\theta'_{b\!f}\prm{\kappa}$ and
$\varphi'_{b\!f}\prm{\kappa}$ equal to $\theta_{b\!f}\prm{\kappa}$ and
$\varphi_{b\!f}\prm{\kappa}$ respectively, where $\sym{i}$ is replaced by
$0$.
% and
% $\varphi'_{b\!f}\prm{\kappa}$ equals $\varphi_{b\!f}\prm{\kappa}$ where
% $\sym{i}$ is replaced by $0$.
The second triple $(e,\theta_{ce}\prm{\kappa},\varphi_{ce}\prm{\kappa})$ of
the template generates the successor node labelled with a program state
$(e,\theta'_{ce}\prm{\kappa},\varphi'_{ce}\prm{\kappa})$ computed
analogously using the composition. The symbolic execution then continues
from the locations $f$ and $e$ in parallel using the classic symbolic
execution. The resulting compact symbolic execution tree is depicted in
Figure~\ref{fig:linSrch}(c). Observe that the two nodes introduced during
template instantiation are drawn with different shape than the others.
% , and their labels describe all execution paths between entry
% and exit nodes corresponding to $p$ and $q$.
Moreover, labels of these nodes immediately indicate all paths in the
flowgraph whose execution is replaced by the application of the template.

If we compare trees at Figures~\ref{fig:linSrch}(b) and
\ref{fig:linSrch}(c), we immediately see that the compact tree is much
smaller than the classic one. In particular, the infinite path in the
classic tree (highlighted by the grey region) does not appear in the compact
one. However, both trees keep the same information in all their leaves. For
example, the program state of the left leaf of the compact tree contains the
following path condition
\[
  \varphi\prm{\kappa}~=~\kappa\geq 0~\wedge~\forall\tau(0\le\tau<\kappa
  \,\rightarrow\,(\tau<\sym{n}\,\wedge\,\sym{A}(\tau)\neq\sym{x}))
  ~\wedge~\kappa<\sym{n}~\wedge~\sym{A}(\kappa)=\sym{x}.
\]
% \[
% \varphi\prm{\kappa}~=~\kappa\ge0~\wedge~ \forall\tau(0\le\tau<\kappa\implies
% (\tau<\sym{n}\,\wedge\,\sym{A}(\tau)\neq\sym{x}))~\wedge~\kappa<\sym{n}
% ~\wedge~\sym{A}(\kappa)=\sym{x}.
% \]
Let us mark all leaves on the left-hand side of the classic tree as
$g_0,g_1,g_2,\ldots$ and let $\varphi_0,\varphi_1,\varphi_2,\ldots$ be
the corresponding path conditions (remember, that each $\varphi_i$ is
the conjunction of labels along the corresponding paths in the tree)
and check that $\varphi_i$ is equivalent to $\varphi\prm{i}$ for each
$i\ge 0$. For example, for $i = 1$ we have
\begin{align*}
  \varphi_1~&=~0<\sym{n}~\wedge~\sym{A}(0)\neq\sym{x}
  ~\wedge~1<\sym{n}~\wedge~\sym{A}(1)=\sym{x},\\
  \varphi\prm{1}~&=~1\geq 0~\wedge~\forall\tau(0\le\tau<1
  \,\rightarrow\,(\tau<\sym{n}\,\wedge\,\sym{A}(\tau)\neq\sym{x}))
  ~\wedge~1<\sym{n}~\wedge~\sym{A}(1)=\sym{x},
\end{align*}
%\[
%\setlength{\arraycolsep}{0pt}
%\begin{array}{rcl}
%  \varphi_1&~=~
%  &0<\sym{n}~\wedge~\sym{A}(0)\neq\sym{x}
%   ~\wedge~1<\sym{n}~\wedge~\sym{A}(1)=\sym{x}\\[2mm]
%  \varphi\prm{1}&\equiv&\forall\tau(0\le\tau<1\implies
%  (\tau<\sym{n}~\wedge~\sym{A}(\tau)\neq\sym{x}))~\wedge\\
%  &\wedge&1<\sym{n}~\wedge~\sym{A}(1)=\sym{x}
%\end{array}
%\]
and hence $\varphi_1\equiv\varphi\prm{1}$. Similarly, each symbolic memory
of a node $g_i$ is an instance $\theta\prm{i}$ of the parametrized symbolic
memory in the left leaf of the compact tree. Analogous relations hold for
leafs on the right-hand sides of the compact and the classic symbolic
execution trees.

%}}}
%{{{ Description of the technique

\section{Description of the Technique}\label{sec:technique}

This section describes the compact symbolic execution in details. For
simplicity, we consider only programs represented by a single flowgraph
manipulating integer variables and read-only integer arrays. The technique
can be extended to handle mutable integer arrays, other data types, and
function calls.

  %{{{ Preliminaries

\subsection{Preliminaries}

Besides the terms and notation introduced in the previous section, we use
also the following terms and notation.

We write $\theta\prm{\bld{\kappa}}$ to emphasise that $\bld{\kappa}$ is the
set of parameters appearing in the symbolic memory $\theta$. Similarly, we
write $\varphi\prm{\bld{\kappa}}$ to emphasise that $\bld{\kappa}$ is the
set of parameters with free occurrences in the formula $\varphi$. We also
write $s\prm{\bld{\kappa}}$ or $(l,\theta,\varphi)\prm{\bld{\kappa}}$, if
$s=(l,\theta\prm{\bld{\kappa}},\varphi\prm{\bld{\kappa}})$.

A \emph{valuation} of parameters is a function $\bld{\nu}$ from a finite set
of parameters to non-negative integers.
% Let $\bld{\nu}$ be a valuation of parameters from $\bld{\kappa}$.
By $\theta\prm{\bld{\nu}}$, $\varphi\prm{\bld{\nu}}$, and $s\prm{\bld{\nu}}$
we denote a symbolic memory $\theta\prm{\bld{\kappa}}$, a formula
$\varphi\prm{\bld{\kappa}}$, and a program state $s\prm{\bld{\kappa}}$
respectively, where all free occurrences of each $\kappa\in\bld{\kappa}$ are
replaced by~$\bld{\nu}(\kappa)$. If $\bld{\kappa}=\{\kappa\}$ is a singleton
and $\bld{\nu}(\kappa)=\nu$, we simply write
$\theta\prm{\kappa},\varphi\prm{\kappa},s\prm{\kappa}$ instead of
$\theta\prm{\bld{\kappa}},\varphi\prm{\bld{\kappa}},s\prm{\bld{\kappa}}$ and
$\theta\prm{\nu},\varphi\prm{\nu},s\prm{\nu}$ instead of
$\theta\prm{\bld{\nu}},\varphi\prm{\bld{\nu}},s\prm{\bld{\nu}}$.

If $\theta$ is a symbolic memory and $\varphi$ is a formula or a symbolic
expression, then $\theta\langle\varphi\rangle$ denotes $\varphi$ where all
occurrences of all symbols $\sym{a}$ are simultaneously replaced by
$\theta(\var{a})$, i.e.~by the value of the corresponding variable stored in
$\theta$.

When $\theta_1$ and $\theta_2$ are two symbolic memories, then
$\theta_1\compose\theta_2$ is a \emph{composed symbolic memory} satisfying
$(\theta_1\compose\theta_2)(\var{a})=\theta_1\langle\theta_2(\var{a})\rangle$
for each variable $\var{a}$. Intuitively, the symbolic memory
$\theta_1\compose\theta_2$ represents an overall effect of a code with
effect $\theta_1$ followed by a code with effect $\theta_2$.

%Further, we use the notation $\theta(\cdot)$ in a more general way. It
%always denotes the operation of replacing each (scalar or array) variable
%$\var{a}$ by $\theta(\var{a})$.  $\theta\langle\cdot\rangle$ denotes the
%operation on symbolic expressions or formulae, where all occurrences of all
%symbols $\sym{a}$ are simultaneously replaced by
%$\theta(\var{a})$. Additionally, $\theta_1\compose\theta_2$ denotes
%\emph{composition of symbolic memories} $\theta_1$ and $\theta_2$ satisfying
%$(\theta_1\compose\theta_2)(\var{a})=\theta_1\langle\theta_2(\var{a})\rangle$
%for each variable $\var{a}$. Intuitively, the symbolic memory
%$\theta_1\compose\theta_2$ represents an overall effect of a code with
%effect $\theta_1$ followed by a code with effect $\theta_2$.

% Finally, for vectors $\vec{u}=(u_1,\ldots,u_n)$ and
% $\vec{v}=(v_1,\ldots,v_n)$ we use $\vec{u}\le\vec{v}$ as an abbreviation
% for $u_1\le v_1\wedge\ldots\wedge u_n\le v_n$.

We define \emph{composition of states} $s_1=(l_1,\theta_1,\varphi_1)$ and
$s_2=(l_2,\theta_2,\varphi_2)$ to be the state $s_1\compose
s_2=(l_2,\theta_1\compose\theta_2,\varphi_1\wedge\theta_1\langle\varphi_2\rangle)$. The
composed state corresponds to the symbolic state resulting from symbolic
execution of the code that produced $s_1$ immediately followed by the code
that produced $s_2$.

We often use a dot-notation to denote elements of a program state $s$: $s.l$
denotes its current location, $s.\theta$ denotes its symbolic memory, and
$s.\varphi$ denotes its path condition. Further, if $u$ is a node of a
symbolic execution tree, then $u.s$ denotes the program state labelling $u$
and we write $u.l$, $u.\theta$, and $u.\varphi$ instead of $(u.s).l$,
$(u.s).\theta$, and $(u.s).\varphi$.

% We say that symbolic memories $\theta_1$ and $\theta_2$ are equivalent,
% written $\theta_1\equiv\theta_2$, if, for each variable $\var{a}$, the
% formula $\theta_1(\var{a})=\theta_2(\var{a})$ holds. Further, we write
% $\varphi_1\equiv\varphi_2$ if the two formulae are equivalent in the logical
% sense. Finally, two states $s_1=(l_1,\theta_1,\varphi_1)$ and
% $s_2=(l_2,\theta_2,\varphi_2)$ are \emph{equivalent}, written $s_1\equiv s_2$,
% if $l_1=l_2$, $\theta_1\equiv\theta_2$, and $\varphi_1\equiv\varphi_2$.

Two program states $s_1,s_2$ are \emph{equivalent}, written $s_1\equiv s_2$,
if $s_1.l=s_2.l$, the formula $s_1.\theta(\var{a})=s_2.\theta(\var{a})$
holds for each variable $\var{a}$, and the formulae $s_1.\varphi$ and
$s_2.\varphi$ are equivalent in the logical sense.

Considered integer programs operate in undecidable theories (like Peano
arithmetic). We assume that there is a function $\sat(\varphi)$ that returns
SAT if it can prove satisfiability of $\varphi$, UNSAT if it can prove
unsatisfiability of $\varphi$, and UNKNOWN otherwise.

  %}}}
  %{{{ Templates

\subsection{Templates and Their Computation}

We start with a formal definition of \emph{cycle}, i.e.~a cyclic path with a
specified entry location and exit edges.
\begin{dfn}[Cycle]\label{def:Part}
  Let $(u,e)$ be an edge of a flowgraph $P$, $\pi=e\omega e$ be a cyclic
  path in $P$ such that $ue$ is not a suffix of $\pi$ and all nodes in
  $\omega e$ are pairwise distinct,
% such that $\omega$ does not contain neither the entry location of $f$ nor
% the location $u$,
  and let $X=\{(u_1,x_1),\ldots,(u_n,x_n)\}$ be the set of all edges of $P$
  that do not belong to the path $\pi$, but their start nodes
  $u_1,\ldots,u_n$ lie on $\pi$. Then $C = (\pi,e,X) $ is a \emph{cycle} in
  $P$, the path $\pi$ is a \emph{core} of $C$, $e$ is an \emph{entry
    location} of $C$, all edges in $X$ are \emph{exit edges} of $C$, and
  each location $x_i$ is called an \emph{exit location} of $C$.
\end{dfn}

We emphasise that the core of a cycle is a cyclic path in a graph
sense. Note that a program loop can generate more independent cycles,
e.g.~if the loop contains interal branching or loop nesting
(see Appendix~\ref{sec:LoopsCycles} for more details).

A template for a cycle $(\pi,e,X)$ is a pair $(e,M)$, where $M$ is a set
containing one parametric program state for each exit edge of the cycle.  A
template for a given cycle can be computed by
Algorithm~\ref{alg:template}. The algorithm uses a function
$\texttt{executePath}(P,\rho)$ which applies classic symbolic execution to
instructions on the path $\rho$ in the program $P$ and returns the resulting
symbolic state $(u,\theta,\varphi)$, where $u$ is the last location in
$\rho$.

\begin{algorithm}[!tb]
  \caption{computeTemplate}\label{alg:template}
  \KwIn{a program $P$ and a cycle $(\pi,e,X)$}
  \KwOut{a template $(e,M)$ or {\bf null } (if the computation fails)}
  \BlankLine
  \DontPrintSemicolon
$(e,\theta,\varphi)\aset\texttt{executePath(}P,\pi\texttt{)}$\;
\label{alg:template:begin}
\lIf{\rm $\sat(\varphi)\neq$ SAT}{
  \label{alg:template:satCheck1}
  \Return{\bf null} \;
}
Set $\theta_*\prm{\kappa}(\var{a})=\sym{a}$~~for each array variable $\var{a}$\;
Set $\theta_*\prm{\kappa}(\var{a})=\bot$~~for each integer variable $\var{a}$\;
\label{alg:template:memoryBegin}
% and \mbox{~~~~~~$\theta_*\prm{\kappa}(\var{A})=\sym{A}$ for each
%  array variable $\var{A}$}\\
\Repeat{$\var{change} = \false$}{ \label{alg:template:repeatBegin}
  $\var{change} \aset \false$\;
  \ForEach{{\rm integer variable} $\var{a}$}{
    \label{alg:template:repeat:foreachBegin}
    \If{$\theta_*\prm{\kappa}(\var{a})=\bot$}{
      \label{alg:template:repeatIsBot}
      \If{$\theta(\var{a})=\sym{a}+c$~~{\rm for some constant} $c$}{
        \label{alg:template:repeat:foreach:BranchAritSeq}
        $\theta_*\prm{\kappa}(\var{a}) \aset \sym{a}+\kappa\cdot c$\;
        $\var{change} \aset \true$\;
      }  
      \If{$\theta(\var{a})=\sym{a}\cdot c$~~{\rm for some constant} $c$}{
        \label{alg:template:repeat:foreach:BranchGeomSeq}
        $\theta_*\prm{\kappa}(\var{a}) \aset \sym{a}\cdot c^\kappa$\;
        $\var{change} \aset \true$\;
      }
      \If{\rm $\theta(\var{a})=g$~~for some symbolic expression $g$ such
        \mbox{~~~that $\theta_*\prm{\kappa}(\var{b})\neq\bot$ for each symbol $\sym{b}$ in $g$}}{
        \label{alg:template:repeat:foreach:BranchTemporary}
        $\theta_*\prm{\kappa}(\var{a}) \aset
        \ite(\kappa>0,\theta_*\prm{\kappa-1}\langle g\rangle,\sym{a})$\;
        $\var{change} \aset \true$\;
        \label{alg:template:repeat:foreachEnd}
        \label{alg:template:repeat:foreach:BranchLast}
      }  
      % \If{$\theta(\var{a})=\sym{a}$~~~and \var{a} is an array }{
      %   \label{alg:template:repeat:foreach:BranchImutableArray}
      %   $\theta_*\prm{\kappa}(\var{a}) \aset \sym{a}$\;
      %   $\var{change} \aset \true$\;
      %   \label{alg:template:repeat:foreachEnd}
      %   \label{alg:template:repeat:foreach:BranchLast}
      % }  
      % \oldtodo{pridat pripad pro monotonni zapis do pole}
    }
  }
}
\label{alg:template:repeatEnd}
\lIf{$\theta_*\prm{\kappa}(\var{a})=\bot$ {\rm for some variable} $\var{a}$}{
  \label{alg:template:isMemoryPrecise}
  \Return{\bf null} \;
}
$\varphi_*\prm{\kappa}\aset\kappa\ge 0~\wedge~
\forall\tau(0\le\tau<\kappa\implies\theta_*\prm{\tau}\langle\varphi\rangle)$\;
\label{alg:template:endMemory}
$M\aset\emptyset$\; \label{alg:template:secondPartBegin}
\ForEach{$(u,x) \in X$}{
  \label{alg:template:foreach2begin}
  Let $\rho$ be the prefix of $\pi$ from $e$ to $u$ \;
  \label{alg:template:pathToExit}
  $(x,\theta,\varphi)\aset\texttt{executePath(}P,\rho x\texttt{)}$\;
  \label{alg:template:exitState}
  \lIf{\rm $\sat(\varphi)=$ UNKNOWN}{
    \label{alg:template:satCheck2a}
    \Return{\bf null} \;
  }
  \If{\rm $\sat(\varphi)=$ SAT}{
    \label{alg:template:satCheck2b}
    $M\aset M \cup \{ (x,~\theta_*\prm{\kappa} \compose \theta,~
    \varphi_*\prm{\kappa} \wedge \theta_*\prm{\kappa}
    \langle\varphi\rangle) \}$
    \label{alg:template:extendM}
    \label{alg:template:foreach2end}
  }
}
\Return{$(e,M)$}
\label{alg:template:end}
\end{algorithm}

The first part of the algorithm
(lines~\ref{alg:template:begin}--\ref{alg:template:endMemory}) tries to
derive a parametric symbolic memory $\theta_*\prm{\kappa}$ and a parametric
path condition $\varphi_*\prm{\kappa}$, which together describe the symbolic
state after $\kappa$ iterations over the core $\pi$ of the cycle $C$, for
any $\kappa \geq 0$. Initially, at line~\ref{alg:template:begin}, we compute
the effect of a \emph{single} iteration of the core $\pi$ and then we check
whether the iteration is feasible. If we cannot prove its feasibility, we
stop the template computation and return \textbf{null}.\footnote{It is
  possible that the iteration is feasible and the chosen SMT solver failed
  to prove it. However, as parametric path conditions of the resulting
  template are derived from $\varphi$, it is highly probable that
  applications of the template in compact symbolic execution would also lead
  to failures of the SMT solver. Such a template is useless.}  Otherwise, we
get a symbolic state $(e,\theta,\varphi)$, whose elements $\theta$ and
$\varphi$ form a basis for the computation of $\theta_*\prm{\kappa}$ and
$\varphi_*\prm{\kappa}$.

We compute $\theta_*\prm{\kappa}$ first. As arrays are read-only, we
directly set $\theta_*\prm{\kappa}(\var{a})$ to $\sym{a}$ for each array
variable $\var{a}$. For integer variables, we initialise
$\theta_*\prm{\kappa}$ to an undefined value $\bot$. Then, in the loop at
lines~\ref{alg:template:repeatBegin}--\ref{alg:template:repeatEnd}, we try
to define $\theta_*\prm{\kappa}$ for as many variables as possible. For each
variable $\var{a}$, $\theta_*\prm{\kappa}(\var{a})$ is defined at most
once. Hence, the loop terminates after finite number of iterations. The
value of $\theta_*\prm{\kappa}(\var{a})$ is defined according to the content
of $\theta(\var{a})$ and known values of $\theta_*\prm{\kappa}$. In
particular, the conditions at
lines~\ref{alg:template:repeat:foreach:BranchAritSeq}
and~\ref{alg:template:repeat:foreach:BranchGeomSeq} check if the values of
$\var{a}$ follow an arithmetic or a geometric progression during the
iterations. If they do, we can easily express the exact value of \var{a}
after any $\kappa$ iterations. Note that the case when the value of \var{a}
variable is not changed along $\pi$ at all is a special case of an
arithmetic progression ($c = 0$). Obviously, one can add support for other
kinds of progression. The condition at
line~\ref{alg:template:repeat:foreach:BranchTemporary} covers the case when
each iteration assigns to $\var{a}$ an expression containing only variables
with known values of $\theta_*\prm{\kappa}$. The if-then-else expression
$\ite(\kappa>0,\theta_*\prm{\kappa-1}\langle g\rangle,\sym{a})$ assigned to
$\theta_*\prm{\kappa}(\var{a})$ says that the value of $\var{a}$ after
$\kappa>0$ iterations is given by the value of expression $g$ where each
symbol $\sym{b}$ represents the value of $\var{b}$ at the beginning of the
last iteration and thus it must be replaced by
$\theta_*\prm{\kappa-1}(\var{b})$. The value of $\var{a}$ after 0 iterations
is obviously unchanged, i.e.~$\sym{a}$.

Once we get to line~\ref{alg:template:isMemoryPrecise}, we check whether we
succeeded to define $\theta_*\prm{\kappa}$ for all variables. If we failed
for at least one variable, then we fail to compute a template for $C$ and we
return \textbf{null}. Otherwise, at line~\ref{alg:template:endMemory} we
compute the formula $\varphi_*\prm{\kappa}$ in accordance with the intuition
provided in Section~\ref{sec:basic}.

The second part of the algorithm
(lines~\ref{alg:template:secondPartBegin}--\ref{alg:template:end}) computes
the set $M$ of the resulting template. As we already know from
Section~\ref{sec:basic}, we try to compute one element of $M$ for each exit
edge $(u, x) \in X$. At line~\ref{alg:template:pathToExit} we compute a path
$\rho$ from the entry location $e$ to $u$ (along $\pi$), where we escape
from $\pi$ to the location $x$. The path $\rho x$ is then symbolically
executed. If we fail to decide feasibility of the path, we fail to compute a
template. If the path is feasible, we can escape $\pi$ by taking the exit
edge $(u,x)$. Therefore, only in this case we add a new element to $M$ at
line~\ref{alg:template:extendM}. The structure of the element follows the
intuition given in Section~\ref{sec:basic}.

One can immediately see that the algorithm always terminates. Now we
formulate a theorem describing properties of the computed template
$(e,M)$. The theorem is crucial for proving soundness and completeness of
compact symbolic execution. Roughly speaking, the theorem says that whenever
a node $u$ of the symbolic execution tree of a program $P$ satisfies
$u.l=e$, then the subtree rooted in $u$ has the property that each branch to
a leaf contains a node $w$ such that $w.s$ corresponds to the composition of
$u.s$ and a suitable instance of some program state of the template (L1),
and vice versa (L2). A proof of the theorem can be found in
Appendix~\ref{sec:Proofs}.

\begin{thm}[Template Properties]
  Let $T$ be a classic symbolic execution tree of $P$ and let
  $\big(e,\{(l_1,\theta_1\prm{\kappa},\varphi_1\prm{\kappa}), \ldots,
  (l_n,\theta_n\prm{\kappa},\varphi_n\prm{\kappa})\}\big)$ be a template for a
  cycle $(\pi,e,X)$ in $P$ produced by Algorithm~\ref{alg:template}. Then
  the following two properties hold:
  \begin{itemize}
  \item[(L1)] For each path $\pi = u \omega$ in $T$ leading from a
    node $u$ satisfying $u.l = e$ to a leaf, there is a node $w$ of
    $\omega$, an index $i \in \{ 1, \ldots, n \}$, and an integer $\nu \geq
    0$ such that $w.s \equiv u.s \compose
    (l_i,\theta_i\prm{\nu},\varphi_i\prm{\nu})$.
  \item[(L2)] For each node $u$ of $T$, an index $i \in \{ 1, \ldots, n \}$,
    and an integer $\nu\geq 0$ such that $u.l = e$ and $(u.\varphi
    \wedge u.\theta\ese{\varphi_i\prm{\nu}})$ is satisfiable, there is a
    successor $w$ of $u$ in $T$ such that $w.s \equiv u.s \compose
    (l_i,\theta_i\prm{\nu},\varphi_i\prm{\nu})$.
  \end{itemize}
\end{thm}

%We finish this section by formulating correctness and termination
%theorems for Algorithm~\ref{alg:template}.

% \begin{thm}[Correctness]
% Whenever Algorithm~\ref{alg:template} returns a template $(e,M)$ for a
% passed program $P$ and a cycle $(\pi,e,X)$ in $P$, then the template
% satisfies Definition~\ref{def:LoopTemplate}.
% \end{thm}

% \begin{thm}[Termination]
% Algorithm~\ref{alg:template} terminates for any program $P$ and any
% cycle $(\pi,e,X)$ in $P$.
% \end{thm}

  %}}}
  %{{{ Compact symbolic execution

\subsection{Compact Symbolic Execution}

The compact symbolic execution is formally defined by
Algorithm~\ref{alg:executeSymbolically}. If we ignore the lines marked by
$\Box$, then we get the classic symbolic execution. As we focus on compact
symbolic execution, we describe the algorithm with $\Box$ lines
included. The algorithm gets a program $P$ and a finite set $p$ of templates
resulting from analyses of some cycles in $P$.
Lines~\ref{alg:executeSymbolically:InitialState}--\ref{alg:executeSymbolically:Root}
create an initial program state, insert it into a queue $Q$, and create the
root of a symbolic execution tree $T$ labelled by the state. 

The queue $Q$ keeps all the program states waiting for their processing in
the \textbf{repeat}-\textbf{until} loop (lines
\ref{alg:executeSymbolically:LoopBegin}--\ref{alg:executeSymbolically:LoopEnd}). The
key part of the loop's body begins at
line~\ref{alg:executeSymbolically:GetTemplatesAt}, where we select at most
one template of $p$ with entry location matching the actual program location
$s.l$. Note that there can be more than one template available at $s.l$ as
more cyclic paths can go through the location. We do not put any constraints
in the selection strategy. We may for example choose randomly. Also note
that we may choose none of the templates (i.e.~we select \textbf{null}), if
there is no template in $p$ for location $s.l$ or even if there are such
templates in $p$.
% all those
%templates from the set $p$, whose entry locations match an actual program
%location $l$. If the selection is not empty, we choose one template $(l,M)$
%of the selected templates.
% This single template $t$ is chosen at
% line~\ref{alg:executeSymbolically:ChooseTemplate}.
%We do not put any constraints to the selection strategy. We may for example
%choose randomly.
If a template $t=(s.l,M)$ is selected, then we get a fresh parameter
(line~\ref{alg:executeSymbolically:UnrollingFreshKappa}) and replace the
original parameter in all tuples of $M$ by the fresh one. This replacement
prevents collisions of parameters of already applied templates. The
\textbf{foreach} loop at
lines~\ref{alg:executeSymbolically:Unrolling}--\ref{alg:executeSymbolically:UnrollingEnd}
creates a successor state $s'$ for each program state in $M$.
% (see line~\ref{alg:executeSymbolically:SuccessorInUnrolling}).
If the template selection at
line~\ref{alg:executeSymbolically:GetTemplatesAt} returns \texttt{null}, we
proceed to line~\ref{alg:executeSymbolically:ClassicStep} and compute
successor states of the state $s$ by the classic symbolic execution.  The
successor states with provably satisfiable path conditions are inserted into
the queue $Q$ and into the compact symbolic execution tree $T$ in the
\textbf{foreach} loop at
lines~\ref{alg:executeSymbolically:AddSATSuccessors}--\ref{alg:executeSymbolically:CreateVertex}.
The successor states with provably unsatisfiable path conditions are ignored
as they correspond to infeasible paths. The \textbf{foreach} loop at
lines~\ref{alg:executeSymbolically:LailedLeafBegin}--\ref{alg:executeSymbolically:LailedLeafEnd}
handles the successor states with path conditions for which we are unable to
decide satisfiability; these states are inserted into the resulting tree $T$
as so-called \emph{failed leaves}. A presence of a failed leaf in the
resulting tree indicates that applied symbolic execution has failed to
explore whole path-space of the executed program.  We do not continue
computation from these states as there is usually a plethora of other states
with provably satisfiable path conditions.

\begin{algorithm}[!t]
    \caption{\texttt{executeSymbolically}} \label{alg:executeSymbolically}
    \KwIn{a program $P$ to be executed 
        \KwInSep \hspace{-1.1cm}\lmark{$\Box$}\hspace{0.97cm}
         and a finite set $p$ of templates computed for cycles in $P$
    }
    \KwOut{a symbolic execution tree $T$ of $P$ (compact tree in $\Box$--version)}
    \BlankLine
    \DontPrintSemicolon
    $s_0 \aset (\text{the starting location of }P,\theta_I,\true)$\;        
    \label{alg:executeSymbolically:InitialState}
    Let $Q$ be a queue of states initially containing only $s_0$ \;
        \label{alg:executeSymbolically:InitialQ}
    Insert the root node labelled by $s_0$ to the empty tree $T$\;
        \label{alg:executeSymbolically:Root}
    \Repeat{\rm $Q$ becomes empty}{ \label{alg:executeSymbolically:LoopBegin}
        Extract the first state $s$ from $Q$ \;
        \If{\rm $s.l$ is either an exit from $P$ or an error location}{
          \textbf{continue}
        }
            $S \aset \emptyset$ \;
            \lmark{$\Box$}
            $t \aset \texttt{chooseTemplate(}s.l,p\texttt{)}$ \;
                \label{alg:executeSymbolically:BoxesBegin}
                \label{alg:executeSymbolically:GetTemplatesAt}
                \label{alg:executeSymbolically:ChooseTemplate}
            \lmark{$\Box$}
            \If {\rm $t \neq $ \bf{null}}{
                \lmark{$\Box$}
                Let $M$ be the second element of $t$, i.e.~$t=(s.l,M)$\; 
                \lmark{$\Box$}
                $\kappa \aset \texttt{getFreshParam()}$ \;
                    \label{alg:executeSymbolically:UnrollingFreshKappa}
                \lmark{$\Box$}
                Replace all occurrences of the former parameter in
                    $M$ by $\kappa$\;
                \lmark{$\Box$}
                \ForEach{$(l,\theta\prm{\kappa},\varphi\prm{\kappa})\in M$}{
                        \label{alg:executeSymbolically:Unrolling}
                    \lmark{$\Box$}
                    $s' \aset s
                        \compose(l,\theta\prm{\kappa},\varphi\prm{\kappa})$\;
                    \label{alg:executeSymbolically:SuccessorInUnrolling}
                    \lmark{$\Box$}
                    Insert $s'$ into $S$\;
                        \label{alg:executeSymbolically:UnrollingEnd}
                }
            }
            \lmark{$\Box$}
            \Else(~~/* apply classic symbolic execution step */){
                    \label{alg:executeSymbolically:BoxesEnd}
                $S \aset \texttt{computeClassicSuccessors(}P,s\texttt{)}$ \;
                    \label{alg:executeSymbolically:ClassicStep}
            }
            Let $u$ be the leaf of $T$ whose label is $s$ \;
                    \label{alg:executeSymbolically:GetLeaf}
            \ForEach{\rm state $s' \in S$ such that $\sat(s'.\varphi)=$ SAT}{
                    \label{alg:executeSymbolically:AddSATSuccessors}
                    Insert $s'$ at the end of $Q$ \; 
                    Insert a new node $v$ labelled with $s'$ and 
                    a new edge $(u,v)$ into $T$ \;
                    \label{alg:executeSymbolically:CreateVertex}
            }
            \ForEach{\rm state $s'\in S$ such that $\sat(s'.\varphi)=$ UNKNOWN}{
                    \label{alg:executeSymbolically:LailedLeafBegin}
                    Insert a new node $v$ labelled with $s'$ and 
                    a new edge $(u,v)$ into $T$ \;
                    Mark the node $v$ in $T$ as a failed leaf
                    \label{alg:executeSymbolically:LailedLeafEnd}
            }
    } \label{alg:executeSymbolically:LoopEnd}
    \Return{$T$} \;
\end{algorithm}

%Observe, that we may apply an optimisation at lines
%\ref{alg:executeSymbolically:AddSATSuccessors} and
%\ref{alg:executeSymbolically:LailedLeafBegin}, where we check for
%satisfiability of $s'.\varphi$ only if it differs from $s.\varphi$.

We finish this section by soundness and completeness theorems for compact
symbolic execution. We assume that $T$ and $T'$ are classic and compact
symbolic execution trees of the program $P$ computed by
Algorithm~\ref{alg:executeSymbolically} without and with $\Box$-lines
respectively. The theorems hold on assumption that our $\sat(\varphi)$
function never returns UNKNOWN, i.e.~neither $T$ nor $T'$ contains failed
leaves. Proofs of both theorems are in Appendix~\ref{sec:Proofs}.

\begin{thm}[Soundness]
  For each leaf node $e \in T$ there is a leaf node $e' \in T'$ and a
  valuation $\bld{\nu}$ of parameters in $e'\!.s$ such that $e.s \equiv
  e'\!.s\prm{\bld{\nu}}$.
\end{thm}

\begin{thm}[Completeness]
For each leaf node $e' \in T'$ there is a leaf node $e \in T$ and
a valuation $\bld{\nu}$ of parameters in $e'\!.s$ such that $e.s \equiv
e'\!.s\prm{\bld{\nu}}$.
\end{thm}

Note that in both theorems we discuss only the relation between all
\emph{finite} branches of the trees $T$ and $T'$. Some infinite branches of
$T$ (like the one in Figure~\ref{fig:linSrch}(b)) corresponding to infinite
iterations along a cyclic path need not be present in $T'$. As symbolic
execution is typically used to cover as many reachable program locations as
possible, missing infinite iterations along cyclic paths can be seen as a
feature rather than a drawback.

  %}}}
  %{{{ Extensions

%\subsection{Extensions of the technique}
%
%\todo{zvazit kompletni zruseni teto podsekce! MT: Ja jsem pro!}
%For clarity of presentation, the compact symbolic execution has been
%described in a very simple settings of programs working only with integer
%variables, read-only integer arrays, and without function calls. However,
%one can extend the technique to programs with more data types, function
%calls allowing recursion etc. Further, one can improve the presented
%technique for computing templates (e.g.~by addition of new rules) or suggest
%a completely new algorithms for template construction. The definition of
%templates can be also extended. For example, one can also consider templates
%with more than one parameter.
%
%%In~\cite{Trt12} we present a version of this technique working with function
%%calls. Moreover, we introduce there a notion of paired templates that can
%%capture some kinds of (unbounded) recursion.

%}}}

%}}}
%{{{ Experimantal results

\section{Experimental Results}\label{sec:experiments}

\subsubsection{Implementation} We have implemented both classic and compact
symbolic execution in an experimental tool called \textsc{rudla}. The tool
uses our ``library of libraries'' called \textsc{bugst} available at
\textsc{SourceForge}~\cite{BUGST}. The sources of \textsc{Rudla} and all
benchmarks mentioned below are available in the same repository. The
implementation also uses \textsc{clang} 2.9~\cite{CLANG}, \textsc{LLVM}
3.1~\cite{LLVM}, and \textsc{Z3} 4.3.0~\cite{Z3}. 

\subsubsection{Evaluation Criteria}
We would like to empirically evaluate and compare the effectiveness of the
classic and compact symbolic execution in exploration of program
paths. Unfortunately, we cannot directly compare explored program paths or
nodes in the constructed trees as a path or a node in a compact symbolic
execution tree have a different meaning than a path or a node in a classic
symbolic execution tree. To compare the techniques, we fix an exploration
method of the trees, namely we choose the breadth-first search as indicated
in Algorithm~\ref{alg:executeSymbolically}, and we measure the time needed
by each of the techniques to reach a particular location in an analysed
program. Note that for compact symbolic execution we also have to fix a
strategy for template selection since there can generally be more than one
template related to one program location. We always choose randomly between
candidate templates.

\subsubsection{Benchmarks and Results}
We use two collections of benchmarks. The first collection contains 13
programs with a marked target location. As our technique is focused on path
explosion caused by loops, all the benchmarks contain typical program loop
constructions. There are sequences of loops, nested loops and also loops
with internal branching. They are designed to produce a huge number of
execution paths. Thus they are challenging for symbolic execution. The
target location is chosen to be difficult to reach.  The first ten
benchmarks have reachable target locations, while the last three do not. For
these three benchmarks, all the execution paths must be explored to give an
answer.

Experimental results of both compact and classic symbolic executions
% on the benchmarks
are presented in Table~\ref{tab:Expriments}. The high numbers of cycles are
due to our translation from LLVM (see Appendix~\ref{sec:ManyDetected} for
more details). The discrepancy between the numbers of detected cycles and
computed templates is mainly due to infeasability of many cycles (see
line~\ref{alg:template:satCheck1} of Algorithm~\ref{alg:template}).

 We want to highlight the
following observations. First, classic symbolic execution was faster
only for benchmarks \texttt{Hello} and \texttt{decode\_packets}. Second, the
number of states visited by the compact symbolic execution is often several
orders of magnitude lower than the number of states visited by the classic
one. At the same time we recall that the semantics of a state in classic and
compact symbolic execution are different.
%, while getting closer to the target
% location sooner.
Finally, presence of quantifiers in path conditions of compact symbolic
executions puts high requirements on skills of the SMT solver. This
leads to SMT failures, which are not seen in classic symbolic execution.

Algorithm~\ref{alg:executeSymbolically} saves SMT failures in
the form of failed leaves in the resulting compact symbolic execution
tree. Therefore, we may think about subsequent analyses for these
leaves. For example, in a failed leaf we may instantiate parameters $\kappa$
by concrete numbers. The resulting formulae will become quantifier-free and
therefore potentially easier for an SMT solver. This way we might be able to
explore paths below the failed leaves. But basically, analyses of failed
leaves are a topic for our further research. Moreover, as SMT solvers are
improving quickly, we may expect that counts of the failures will decrease
over time.

\begin{table*}[!t]
  \begin{center}
    \setlength{\tabcolsep}{4pt}
    \renewcommand\arraystretch{1.06}
    \begin{tabular}{|c||c|c|c||c|c|c||c|c||}
      \hline
      & \multicolumn{3}{c||}{\textbf{Templates}}
      & \multicolumn{3}{c||}{\textbf{Compact SE}}
      & \multicolumn{2}{c||}{\textbf{SE}} \\
      \cline{2-9}
      \textbf{Benchmark} &
       Time & Count & Cycles &
       Time & States & SMTFail &
       Time & States \\
       \hline
      %   hello          & 12.3 & 2 & 126(2) & 2.3   & 187   & 0  & 4.5 & 2262\\
      %   HW             & 31.9 & 4 & 252(4) & 45.4  & 1048  & 4  & T/O & 223823\\
      %   HWM            & 48.1 & 5 & 336(5) & T/O   & 5125  & 24 & T/O & 162535\\
      %   matrIR         & 4.2  & 4 & 28(4)  & 82.9  & 1234  & 6  & T/O & 270737\\
      %   matrIR\_dyn    & 14.8 & 5 & 30(5)  & 240.5 & 2472  & 13 & T/O & 267636\\
      %   VM             & 8.6  & 6 & 64(6)  & T/O   & 2274  & 64 & T/O & 205577\\
      %   VMS            & 4.2  & 3 & 32(3)  & 5.4   & 466   & 0  & 99.8 & 281263\\
      %   decode\_packets& 18.3 & 5 & 26(5)  & 39.9  & 1276  & 0  & 16.3 & 8992 \\
      %   WinDriver      & 17.8 & 5 & 26(5)  & 59.2  & 1370  & 1  & T/O & 206903\\
      %   EQCNT          & 12.2 & 3 & 12(3)  & 10.6  & 345   & 0  & T/O & 179803\\
      % \hline
      %   EQCNTex        & 5.8  & 4 & 24(4)  & T/O   & 10581 & 0  & T/O & 251061\\
      %   OneLoop        & 0.1  & 1 & 2(1)   & 0.1   & 41    & 0  & T/O & 38230\\
      %   TwoLoops       & 0.3  & 2 & 4(2)   & 0.1   & 25    & 0  & T/O & 917343\\
        hello          & 12.3 & 2 & 126 & 2.3   & 187   & 0  & 4.5 & 2262\\
        HW             & 31.9 & 4 & 252 & 45.4  & 1048  & 4  & T/O & 223823\\
        HWM            & 48.1 & 5 & 336 & T/O   & 5125  & 24 & T/O & 162535\\
        matrIR         & 4.2  & 4 & 28  & 82.9  & 1234  & 6  & T/O & 270737\\
        matrIR\_dyn    & 14.8 & 5 & 30  & 240.5 & 2472  & 13 & T/O & 267636\\
        VM             & 8.6  & 6 & 64  & T/O   & 2274  & 64 & T/O & 205577\\
        VMS            & 4.2  & 3 & 32  & 5.4   & 466   & 0  & 99.8 & 281263\\
        decode\_packets& 18.3 & 5 & 26  & 39.9  & 1276  & 0  & 16.3 & 8992 \\
        WinDriver      & 17.8 & 5 & 26  & 59.2  & 1370  & 1  & T/O & 206903\\
        EQCNT          & 12.2 & 3 & 12  & 10.6  & 345   & 0  & T/O & 179803\\
      \hline
        EQCNTex        & 5.8  & 4 & 24  & T/O   & 10581 & 0  & T/O & 251061\\
        OneLoop        & 0.1  & 1 & 2   & 0.1   & 41    & 0  & T/O & 38230\\
        TwoLoops       & 0.3  & 2 & 4   & 0.1   & 25    & 0  & T/O & 917343\\
      \hline
      \hline
        \textbf{Total time} &
        \multicolumn{3}{c||}{240} &
        \multicolumn{3}{c||}{1800} &
        \multicolumn{2}{c||}{3900} \\
      \hline
    \end{tabular}
  \end{center}
  \caption{Experimental results of compact and classic symbolic  
    executions. The compact symbolic execution approach is divided into 
    computation of
    templates and building of compact symbolic execution tree.
    All the times are in seconds, where 'T/O' identifies
    exceeding 5 minutes timeout. 'Count' represents the number of computed
    templates, 'Cycles' shows the number of detected cycles.
    % in form ``all(symbolically-executable)" (see
    % Appendix~\ref{sec:ManyDetected} of~\cite{TECHREP} for more details).
    'SMTFail' represents the number of failed SMT queries. There was no SMT
    failure during classic SE of our benchmarks.}
  \label{tab:Expriments}
\end{table*}

\begin{table*}[!t]
  \begin{center}
    \setlength{\tabcolsep}{4pt}
    \renewcommand\arraystretch{1.1}
    \begin{tabular}{|c||c|c|c|c|c|c|}
      \hline
         &{\textbf{Time}}&{\textbf{safe}}&{\textbf{unsafe}}&
         {\textbf{timeout}}&{\textbf{unsupported}}&{\textbf{points}} \\
      \hline\hline
      {\textbf{Compact SE}}& 300+4920 &  21  &  25  &   15    &   13+5 & 67\\
      \hline
      {\textbf{SE}}        &   8700   &  10  &  27  &   28    &   13+1 & 47 \\
      \hline
    \end{tabular}
  \end{center}
  \caption{Experimental results of compact and classic symbolic
    executions on 79 SV-COMP 2013 benchmarks in the category 'loops'. Time is
    in seconds. For compact SE we provide template 
    computation time plus execution time. 'safe' and 'unsafe' report the numbers of
    programs where the tool decides unreachability and reachability of a
    marked error location, respectively (all these answers are correct).
    'timeout' presents the number of symbolic executions exceeding 5 minutes.
    'unsupported' represents the number of compilation failures plus
    failures during an analysis. 'points' shows the number of points the
    tools would get according to the SV-COMP 2013 rules.}
\label{tab:ExprimentsSVCOMP}
\end{table*}

The second collection of benchmarks is the whole category 'loops' taken
from SV-COMP 2013 (revision 229)~\cite{B13}. The results are depicted in
Table~\ref{tab:ExprimentsSVCOMP}.

All the presented experiments were done on a laptop Acer Aspire 5920G (2
$\times$ 2GHz, 2GB) running Windows 7 SP1 64-bit.

%To get a different but rather marginal view of effectiveness of
%compact symbolic execution, we also include in
%Table~\ref{tab:Expriments} run-times of the well known symbolic
%execution tool \textsc{Pex}~\cite{TdH08} on our benchmarks. We
%measured time, till the target location was covered by a test or till
%the tool completed its search. Of course, the comparison with \Pex is
%not accurate, since it does not explore execution paths in
%breath-first manner.

%}}}
%{{{ Related work

\section{Related Work}\label{sec:related}

The symbolic execution was introduced by King in 1976~\cite{Kin76}.
% , where the author introduced the original concept of classic symbolic
% execution. Nevertheless, performance issues related to the path explosion
% problem were not tackled.
The original concept was generalised in~\cite{KPV03} for programs with heap
by introducing lazy initialisation of dynamically allocated data structures.
The lazy initialisation algorithm was further improved and formally defined
in~\cite{DLR11}. Another generalisation step was done in~\cite{KS05}, where
the authors attempt to avoid symbolic execution of library code (called from
an analysed program), since such code can be assumed as well defined and
properly tested.

In~\cite{SPmCS09,GL11}, the path explosion problem is tackled by focusing on
program loops. The information inferred from a loop allows to talk about
multiple program paths through that loop. But the goal is to explore classic
symbolic execution tree in some effective manner: more interesting paths
sooner. Approaches~\cite{G07,AGT08} share the same goal as the previous
ones, but they focus on a computation of function summaries rather than on
program loops. 

Our goal is completely different: instead of guiding exploration of paths in
a classic symbolic execution tree, we build a tree that keeps the same
information and contains less nodes.
% has less branches as one leaf of our compact tree can
% represent potentially infinitely many leaves of the clasic tree. 
%
% acollapse possibly infinitely many subpaths (with regularities in states
% along them) into a single node in compact symbolic execution tree.
%
%Paths in compact symbolic execution tree can thus be exponentially
%shorter than corresponding paths in the classic one.
% The main disadvantage of the approaches \cite{SPmCS09,GL11,G07,AGT08} over
% ours is the exploration of paths of \emph{classic} symbolic execution
% tree. These paths can be exponentially longer then corresponding ones in
% compact symbolic execution tree. The same negative effect therefore
% applies to lengths of the corresponding path conditions of these paths.
%
In particular, templates of compact symbolic execution have a different
objective than summarisation used in~\cite{G07,AGT08,GL11}. While
summarisation basically caches results of some finite part of symbolic
execution for later fast reuse, our templates are supposed to \emph{replace}
potentially infinite parts of symbolic executions by a single node.

Techniques~\cite{QNR11,SH10} group paths of classic symbolic execution
tree according to their effect on symbolic values of a priori given
output variables, and explore only one path per group. We consider all
program variables and we explore all program paths (some of them are
explored simultaneously using templates).
%Techniques~\cite{QNR11,SH10} group paths of classic symbolic
%execution tree according to an a priori given constraints on a program
%output, while exploring only one path per group. We do not require any such
%constraints.

Finally, in our previous work~\cite{ST12} we compute a non-trivial necessary
condition for reaching a given target location in a given program. In other
words, the result of the analysis is a first order logic formula. In the
current paper, we focus on a fast exploration of as many execution paths as
possible. The technique produces a compact symbolic execution tree. Note
that we do not require any target location, since we do not focus on a
program location reachability here. Nevertheless, to achieve our goal, we
adopted a part of a technical stuff introduced in~\cite{ST12}. Namely, lines
\ref{alg:template:memoryBegin}--\ref{alg:template:repeatEnd} of
Algorithm~\ref{alg:template} are similar to the computation of a so-called
iterated memory, which is in~\cite{ST12} an over-approximation of the memory
content after several iterations in a program loop. In the current
technique, the memory content must always be absolutely precise. Moreover,
here we analyse flowgraph cycles while~\cite{ST12} summarises program loops.

% only parts of program loops, which is not the case for~\cite{ST12}.

% In this work we do not require any target location, as we do not
%focus on program reachability. Instead, we focus on exploration on all
%finite program paths in a \emph{very} compact manner. To achieve this
%goal, nevertheless, we adopted a part of a technical stuff introduced
%in~\cite{ST12}. Namely, lines
%\ref{alg:template:memoryBegin}--\ref{alg:template:repeatEnd} of
%Algorithm~\ref{alg:template} are simplified version of a computation
%of so-called iterated memory, which is in~\cite{ST12} an
%over-approximation of memory content after several iterations in a
%program loop. In this technique the computed memory must be absolutely
%precise and we may also analyse only parts of program loops.
%
%Finally, the idea of declarative parametric description of the effect
%of a program loop has been already used in our paper~\cite{ST12}. In
%contrast to this paper, the goal of~\cite{ST12} is to find a necessary
%condition on input values to reach a given program location. While
%compact symbolic execution can work with templates generated by
%various techniques (one of them is the presented analysis of program
%cycles), \cite{ST12} focuses on a declarative parametrized description
%of program loops. However, program loops considered in \cite{ST12} can
%have more paths around the loop (generated e.g.~by branching statement
%in the loop).

%}}}
%{{{ Conclusion 

\section{Conclusion}\label{sec:conclusion}

We have introduced a generalisation of classic symbolic execution, called
compact symbolic execution. Before building symbolic execution tree, the
compact symbolic execution computes \emph{templates} for cycles of an
analysed program. A template is a parametric and declarative description of
the overall effect of a related cycle. Our experimental results indicate
that the use of templates during the analysis leads to faster exploration of
program paths in comparison with the exploration speed of classic symbolic
execution. Also a number of symbolic states computed during the program
analysis is considerably smaller. On the other hand, compact symbolic
execution constructs path conditions with quantifiers, which leads to more
failures of SMT queries.

\bibliographystyle{plain}
%\bibliography{cse_atva2013}
%\input{cse_atva2013_arXiv.bbl}

\newpage
\appendix

%{{{ Program Loops and Cycles

\section{Cycles and Program Loops}
\label{sec:LoopsCycles}

We illustrate the difference between cycles formally defined in
Definition~\ref{def:Part} and loop constructs of programming languages using
two short examples. The examples are instances of two common loop
structures: Figure~\ref{fig:cse:LoopsCycles1} shows a loop with an internal
branching and Figure~\ref{fig:cse:LoopsCycles2} presents a code with two
nested program loops. According to our definition of a cycle, the core of a
cycle is a single cyclic path in the graph sense (satisfying some additional
conditions). The flowgraph of Figure~\ref{fig:cse:LoopsCycles1} contains
four cycles while the flowgraph of Figure~\ref{fig:cse:LoopsCycles2} has
even seven cycles. Cores of these cycles are listed in captions of the
figures. One can immediately see that there is no one-to-one correspondence
between loops in a source code and cycles in the flowgraphs.

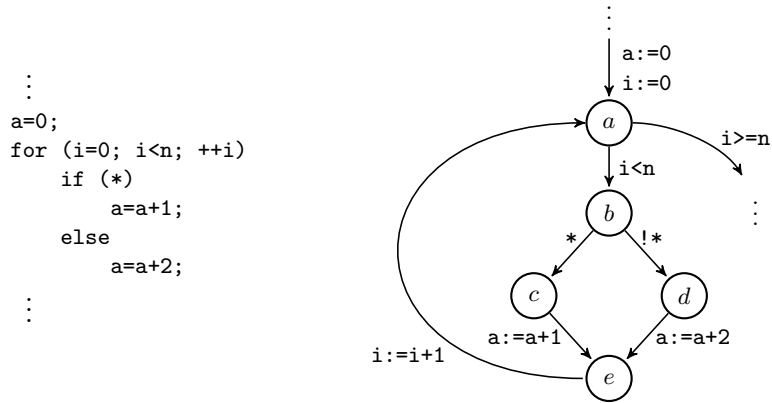
\begin{figure}[hb]
  \begin{center}
    \begin{tabular}{ccc}
      \begin{tabular}[b]{l}
        \texttt{~\vdots}\\
        \texttt{a=0;}\\
        \texttt{for (i=0; i<n; ++i)}\\
        \texttt{~~~~if (*)}\\
        \texttt{~~~~~~~~a=a+1;}\\
        \texttt{~~~~else}\\
        \texttt{~~~~~~~~a=a+2;}\\
        \texttt{~\vdots}\\[5ex]
        ~
      \end{tabular}
      & ~~~~~ &
      \begin{tabular}[b]{c}
%        \includegraphics[width=2.5cm]{fig_cse_icmp_rep}
%        \centering
        \tikzstyle{loc} = [circle,thick,draw,minimum size=6mm] %,inner sep=1pt]
        \tikzstyle{pre} = [<-,shorten <=1pt,>=stealth',semithick]
        \tikzstyle{post} = [->,shorten <=1pt,>=stealth',semithick]
        \begin{tikzpicture}[node distance=1.1cm]   
          \node [] (0) {\vdots};
          \node [loc] (a) [below of=0,yshift=-4mm] {$a$}
            edge [pre] node [label=right:
            \begin{tabular}{l}
              \texttt{a:=0}\\\texttt{i:=0}
            \end{tabular}] {} (0);
          \node [loc] (b) [below of=a,yshift=-1mm] {$b$}
            edge [pre] node [label=right:\texttt{i<n}] {} (a);
          \node [loc] (c) [below of=b,xshift=-10mm] {$c$}
            edge [pre] node [label=above:\texttt{*}] {} (b);
          \node [loc] (d) [below of=b,xshift=10mm] {$d$}
            edge [pre] node [label=above:~\texttt{!*}] {} (b);
          \node [loc] (e) [below of=c,xshift=10mm] {$e$}
            edge [pre] node [label=left:\texttt{a:=a+1}] {} (c)
            edge [pre] node [label=right:\texttt{a:=a+2}] {} (d)
            edge [post,out=180,in=180,looseness=2.5]
              % node [label=right:\texttt{i:=i+1}] {} (a);
              node [overlay,near start,below] {\texttt{i:=i+1}~~~~~~~~~~} (a);
           \node [] (00) [below of=a,xshift=20mm] {\vdots~~}
            edge [pre,in=0,out=120,looseness=0.8] node [label=right:~~\texttt{i>=n}] {} (a);
        \end{tikzpicture}
      \end{tabular}
    \end{tabular}
  \end{center}
  \caption{A source code and a flowgraph of a single program loop with an
    internal branching. The symbol \texttt{*} represents any branching
    condition. The flowgraph has four independent cycles with cores
    \textit{abcea}, \textit{abdea}, \textit{eabce}, and \textit{eabde}.}
  \label{fig:cse:LoopsCycles1}
\end{figure}

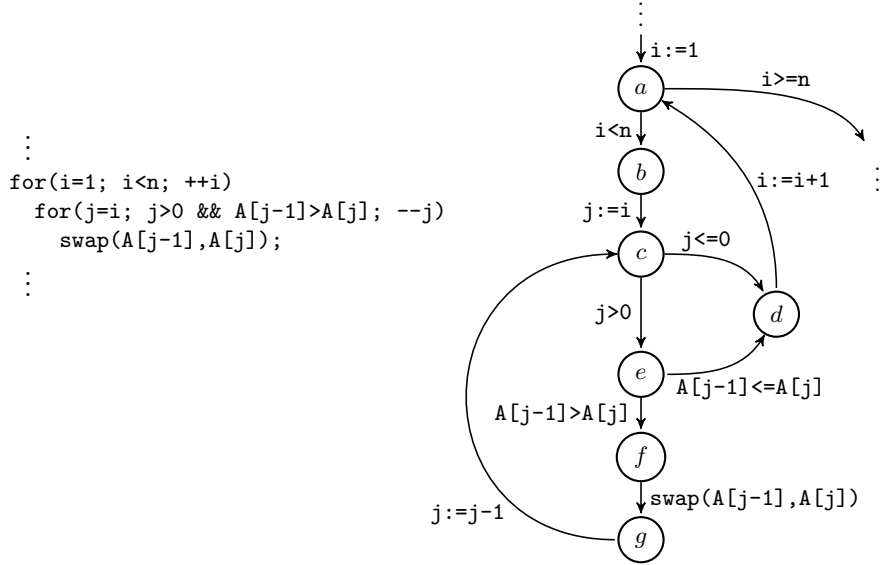
\begin{figure}[!htb]
  \begin{center}
    \begin{tabular}{cc}
      \begin{tabular}[b]{l}
        \texttt{~\vdots}\\
        \texttt{for(i=1; i<n; ++i)}\\
        \texttt{~~for(j=i; j>0 \&\& A[j-1]>A[j]; --j)}\\
        \texttt{~~~~swap(A[j-1],A[j]);}\\
        \texttt{~\vdots}\\[23.3ex]
        ~
      \end{tabular}
      \begin{tabular}[b]{c}
        \tikzstyle{loc} = [circle,thick,draw,minimum size=6mm] %,inner sep=1pt]
        \tikzstyle{pre} = [<-,shorten <=1pt,>=stealth',semithick]
        \tikzstyle{post} = [->,shorten <=1pt,>=stealth',semithick]
        \begin{tikzpicture}[node distance=1.1cm]
          % \node [] (0) {\vdots};
          % \node [loc] (a) [below of=0] {$a$}
          %   edge [pre] node [label=right:\texttt{i:=1}] {} (0);
          % \node [loc] (b) [below of=a,xshift=-18mm] {$b$}
          %   edge [pre,in=180,out=60,looseness=0.8] node [label=above:\texttt{i<n}] {} (a);
          % \node [loc] (c) [below of=b,yshift=-5mm] {$c$}
          %   edge [pre] node [label=right:\texttt{j>0}] {} (b);
          % \node [loc] (d) [below of=a,yshift=-8mm] {$d$}
          %   edge [pre,in=0,out=120] node [label=above:\texttt{j<=0}] {} (b)
          %   edge [pre,in=0,out=240] node [label=below:~~~~~\texttt{A[j-1]<=A[j]}] {} (c)
          %   edge [post] node [label=right:\texttt{i:=i+1}] {} (a);
          % \node [loc] (e) [below of=c] {$e$}
          %   edge [pre] node [label=left:\texttt{A[j-1]>A[j]}] {} (c);
          % \node [loc] (f) [below of=e] {$f$}
          %   edge [pre] node [label=left:\texttt{a:=a+1}] {} (e)        
          %   edge [post,out=180,in=180,looseness=1.7,overlay]
          %     node [overlay,near end,above] {\texttt{j:=j-1}~~~~~~~} (b);
          %  \node [] (00) [below of=a,xshift=18mm,overlay] {\vdots~~}
          %   edge [pre,in=0,out=120,looseness=0.8] node
          %   [label=above:~~\texttt{i>=n}] {} (a);
          \node [] (0) {\vdots};
          \node [loc] (a) [below of=0] {$a$}
            edge [pre] node [label=right:\texttt{i:=1}] {} (0);
          \node [loc] (b) [below of=a] {$b$}
            edge [pre] node [label=left:\texttt{i<n}] {} (a);          
          \node [loc] (c) [below of=b] {$c$}
            edge [pre] node [label=left:\texttt{j:=i}] {} (b);
          \node [loc] (d) [below of=c,yshift=3mm,xshift=18mm] {$d$}
            edge [pre,in=0,out=120] node [label=above:\texttt{j<=0}~~~~~] {} (c)
            edge [post,out=90, in=330] node [label=right:\texttt{i:=i+1}] {} (a);
          \node [loc] (e) [below of=c,yshift=-5mm] {$e$}
            edge [pre] node [label=left:\texttt{j>0}] {} (c)
            edge [post,out=0,in=240] node [label=below:~~~~~~~\texttt{A[j-1]<=A[j]}] {} (d);
          \node [loc] (f) [below of=e] {$f$}
            edge [pre] node [label=left:\texttt{A[j-1]>A[j]}] {} (e);
          \node [loc] (g) [below of=f] {$g$}
            edge [pre] node [label=right:\texttt{swap(A[j-1],A[j])}] {} (f)        
            edge [post,out=180,in=180,looseness=1.8,overlay]
              node [overlay,near start,below] {\texttt{j:=j-1}~~~~~~~~~~} (c);
           \node [] (00) [below of=a,xshift=32mm] {\vdots~~}
            edge [pre,in=0,out=120,looseness=0.8] node
            [label=above:~~\texttt{i>=n}] {} (a);
        \end{tikzpicture}
      \end{tabular}
    \end{tabular}
  \end{center}
  \caption{A source code and a flowgraph of two nested program loops
    (insertion sort). The flowgraph has seven independent cycles with cores
    \textit{abcda}, \textit{abceda}, \textit{cdabc}, \textit{cedabc},
    \textit{cefgc}, \textit{dabcd}, and \textit{dabced}.}
  \label{fig:cse:LoopsCycles2}
\end{figure}

%%%%% STAREJ obrazek s chybama v (d), kreslenej "rukou"
% \begin{figure}[!htb]
% \begin{center}
% \begin{tabular}{ccc}
% \begin{tabular}{l}
% \texttt{a=0;}\\
% \texttt{for (i=0; i<n; ++i)}\\
% \texttt{~~~~if (*)}\\
% \texttt{~~~~~~~~a=a+1;}\\
% \texttt{~~~~else}\\
% \texttt{~~~~~~~~a=a+2;}\\
% \end{tabular}
% & ~~~~~~~~~ &
% \begin{tabular}{l}
% \texttt{for(i=1; i<n; ++i)}\\
% \texttt{~~for(j=i; j>0 \&\& A[j-1]>A[j]; --j)}\\
% \texttt{~~~~swap(A[j-1],A[j]);}\\
% \end{tabular}
% \\ \\
% (a) & & (b)
% \\ \\
% \begin{tabular}{c}
% \includegraphics[width=4cm]{fig_cse_loopscycles_intbranch}
% \end{tabular}
% & &
% \begin{tabular}{c}
% \includegraphics[width=4.5cm]{fig_cse_loopscycles_nested}
% \end{tabular}
% \\
% (c) & & (d)
% \end{tabular}
% \end{center}
% \caption{\marek{(a) A single program loop with internal branching. The symbol
% $*$ represents any branching condition. (b) Two nested program loops
% (insertion sort). (c) A flowgraph corresponding to (a) with four
% independent cycles with cores \textit{abcea}, \textit{abdea},
% \textit{eabce}, and \textit{eabde}. (d) A flowgraph corresponding to
% (b) with seven independent cycles with cores \textit{abda},
% \textit{abcda}, \textit{bdab}, \textit{bcdab}, \textit{bcefb},
% \textit{dabd}, and \textit{dabcd}.}}
% \label{fig:cse:LoopsCycles}
% \end{figure}

%}}}
%{{{ Proofs

\section{Proofs of Theorems}
\label{sec:Proofs}

\setcounter{lem}{0}
\setcounter{thm}{0}

Let $\varphi$, $\varphi'$ be symbolic expressions, $\theta$, $\theta'$ be
symbolic memories, and $s$, $s'$ be program states. Then $\varphi \equiv
\varphi'$, if $\varphi$ and $\varphi'$ are either logically equivalent
formulae or two terms such that $\varphi = \varphi'$ is valid. $\theta
\equiv \theta'$, if for each variable \var{a} we have $\theta(\var{a})
\equiv \theta'(\var{a})$. Finally, recall that $s \equiv s'$, if both $s$
and $s'$ have equal or equivalent components. Now we formulate and prove one
auxiliary lemma that will be used in the subsequent proofs.

\begin{lem}[Equivalent Compositions]
  Let $s$, $s'$, and $s''$ be program states, $\bld{\nu}$ and $\bld{\nu}'$
  be valuations of parameters in $s$ and $s'$ respectively such that
  $\bld{\nu} \cup \bld{\nu}'$ is also a valuation, $\theta$, $\theta'$, and
  $\theta''$ be symbolic memories, and $\varphi$ and $\psi \wedge \psi'$ be
  formulae. The following relations hold:
  \begin{enumerate}
  \item $(\theta \compose \theta')\ese{\varphi} \equiv
    \theta\ese{\theta'\ese{\varphi}}$
  \item $\theta\ese{\psi} \wedge \theta\ese{\psi'} \equiv \theta\ese{\psi
      \wedge \psi'}$
  \item $\theta \compose (\theta' \compose \theta'') \equiv (\theta \compose
    \theta') \compose \theta''$
  \item $s \compose (s' \compose s'') \equiv (s \compose s') \compose s''$
  \item $s\prm{\bld{\nu}} \compose s'\prm{\bld{\nu}'} \equiv (s \compose
    s')\prm{\bld{\nu} \cup \bld{\nu}'}$
  \end{enumerate}
\end{lem}
\begin{proof}
1. The expression $(\theta \compose
 \theta')\ese{\varphi}$ simultaneously substitutes each symbol
 $\sym{a}$ in $\varphi$ by a symbolic expression $(\theta \compose
 \theta') \ese{\theta_I^{-1}(\sym{a})} =
 \theta\ese{\theta'(\theta_I^{-1}(\sym{a}))}$ This follows directly from
 the definition of $\compose$. In the expression
 $\theta\ese{\theta'\ese{\varphi}}$, there we have to apply the
 substitution twice. First we simultaneously substitute each symbol
 $\sym{a}$ in $\varphi$ by a symbolic expression
 $\theta'(\theta_I^{-1}(\sym{a}))$. If the resulting formula contains
 any symbol, then it must necessarily lie in some of the substituted
 expressions $\theta'(\theta_I^{-1}(\sym{a}))$. Therefore, it is
 sufficient to apply the second substitution only to the substituted
 expressions $\theta'(\theta_I^{-1}(\sym{a}))$. In other words, it is
 sufficient to apply only one simultaneous substitution of symbols
 $\sym{a}$ in $\varphi$ by symbolic expressions
 $\theta\ese{\theta'(\theta_I^{-1}(\sym{a}))}$.

2. The equivalence is obviously valid, since the operation
 $\theta\ese{\cdot}$ only applies symbol  substitutions inside formulae
 $\psi$ and $\psi'$.

 3. According to the definition of $\compose$ we have for each variable
 \var{a} the following: $(\theta \compose (\theta' \compose
 \theta''))(\var{a}) \equiv \theta\ese{(\theta' \compose \theta'')(\var{a})}
 \equiv \theta\ese{\theta'\ese{\theta''(\var{a})}} \equiv (\theta \compose
 \theta')\ese{\theta''(\var{a})} \equiv ((\theta \compose \theta') \compose
 \theta'')(\var{a})$.

4. Let $s$,
 $s'$, and $s''$ be program states $(l, \theta, \varphi)$,
 $(l', \theta', \varphi')$, and $(l'', \theta'', \varphi'')$
 respectively. According to the definition of $\compose$, we 
 have $s \compose (s' \compose s'') = (l'', \theta \compose (\theta' \compose
 \theta''),~\varphi \wedge \theta\ese{\varphi' \wedge
 \theta'\ese{\varphi''}})$ and $(s \compose s') \compose s'' = (l'', (\theta \compose
 \theta') \compose \theta'',~(\varphi \wedge \theta\ese{\varphi'}) \wedge
 (\theta \compose \theta')\ese{\varphi''})$. We prove the
 equivalence of the last components. According to points 1.
 and 2. we have $(\varphi \wedge \theta\ese{\varphi'}) \wedge (\theta
 \compose \theta')\ese{\varphi''} \equiv \varphi \wedge
 \theta\ese{\varphi'} \wedge \theta\ese{\theta'\ese{\varphi''}} \equiv
 \varphi \wedge \theta\ese{\varphi' \wedge \theta'\ese{\varphi''}}$.

5. The equivalence follows from
 these two facts: (a) The composition of states operates on symbols,
 while the parameter substitution operates on parameters. (b)
 $\bld{\nu} \cup \bld{\nu}'$ is supposed to be a valuation. Therefore,
 if there is $(\kappa, \nu) \in \bld{\nu}$ and $(\kappa, \nu') \in
 \bld{\nu}'$, then $\nu = \nu'$ must be valid.
\qed
\end{proof}

\begin{thm}[Template Properties]\label{thm:l1l2}
  Let $T$ be a classic symbolic execution tree of $P$ and let
  $\big(e,\{(l_1,\theta_1\prm{\kappa},\varphi_1\prm{\kappa}), \ldots,
  (l_n,\theta_n\prm{\kappa},\varphi_n\prm{\kappa})\}\big)$ be a template for a
  cycle $(\pi,e,X)$ in $P$ produced by Algorithm~\ref{alg:template}. Then
  the following two properties hold:
  \begin{itemize}
  \item[(L1)] For each path $\pi = u \omega$ in $T$ leading from a
    node $u$ satisfying $u.l = e$ to a leaf, there is a node $w$ of
    $\omega$, an index $i \in \{ 1, \ldots, n \}$, and an integer $\nu \geq
    0$ such that $w.s \equiv u.s \compose
    (l_i,\theta_i\prm{\nu},\varphi_i\prm{\nu})$.
  \item[(L2)] For each node $u$ of $T$, an index $i \in \{ 1, \ldots, n \}$,
    and an integer $\nu\geq 0$ such that $u.l = e$ and $(u.\varphi
    \wedge u.\theta\ese{\varphi_i\prm{\nu}})$ is satisfiable, there is a
    successor $w$ of $u$ in $T$ such that $w.s \equiv u.s \compose
    (l_i,\theta_i\prm{\nu},\varphi_i\prm{\nu})$.
  \end{itemize}
\end{thm}
\begin{proof}
% Our goal is to show that all tuples inserted into $M$ at line
% \ref{alg:template:extendM} together satisfy
% Definition~\ref{def:LoopTemplate}. Formulae $\varphi_*\prm{\kappa}
% \wedge \theta_*\prm{\kappa} \langle\varphi\rangle$ constructed at line
% \ref{alg:template:extendM} are trivially satisfiable for $\kappa = 0$,
% since $\theta_*\prm{0}(\var{a}) = \theta_I(\var{a})$ and satisfiability of
% $\varphi$ is ensured at line \ref{alg:template:satCheck2}. Further,
% for any two different  elements of $X$ the function
% \texttt{executePath} returns two formulae $\varphi$ and $\varphi'$
% satisfying $\varphi \wedge \varphi' \equiv \false$ (this is a
% consequence of the same statement for leaves of classic symbolic
% execution tree; see \cite{Kin76} for details). Therefore,
% \begin{align*}
% (&\varphi_*\prm{\kappa} \wedge \theta_*\prm{\kappa}
% \langle\varphi\rangle) \wedge (\varphi_*\prm{\kappa} \wedge
% \theta_*\prm{\kappa} \langle\varphi'\rangle)~\equiv\\
% &\varphi_*\prm{\kappa} \wedge \theta_*\prm{\kappa}\langle\varphi
% \wedge \varphi'\rangle~\equiv\\
% &\varphi_*\prm{\kappa} \wedge \false~\equiv\\
% &\false.
% \end{align*}
% It remains to prove properties (L1) and (L2) of
% Definition~\ref{def:LoopTemplate}.
  We start with (L1). Let $T$ be a classic symbolic execution tree of $P$,
  $(e, \theta, \varphi)$ be a symbolic state computed at line
  \ref{alg:template:begin}, and $(x, \hat{\theta}, \hat{\varphi})$ be a
  symbolic state computed at line \ref{alg:template:exitState} (for some
  exit edge from $X$). Further, let $u$ be a node of $T$ such that $u.l =
  e$, and let $\delta = u_0 \ldots u_1 \ldots u_2$ $\ldots u_{\nu} \ldots w$
  be a path in $T$ starting at $u$ (i.e.~$u_0 = u$), then iterating the core
  $\pi$ exactly $\nu \geq 0$ times (i.e.~all the nodes $u_i \in \delta$ are
  exactly those having $u_i.l = e$), and then $\delta$ finally leaves the
  core $\pi$ by following the path towards the node $w$, satisfying $w.l =
  x$. We use the memory composition operation to express memories of the
  nodes $u_i$ along $\delta$ as follows.
\begin{align*}
u_1.\theta &= u.\theta \compose \theta
\\
u_2.\theta &= u_1.\theta \compose \theta = u.\theta \compose (\theta \compose
\theta)
\\
\cdots
\\
u_{\nu}.\theta &= u_{\nu-1}.\theta \compose \theta = u.\theta \compose
(\underset{\nu}{\underbrace{\theta \compose \cdots \compose \theta}}).
\end{align*}
If we denote the composition of $i$ symbolic memories $\theta$ by
$\theta^i$, where $\theta^0 = \theta_I$, $\theta^1 = \theta$, and
$\theta^i = \theta \compose \theta^{i-1}$, then we have $u_{i}.\theta =
u.\theta \compose \theta^i$ and we get
\begin{align*}
w.\theta = u.\theta \compose (\theta^{\nu} \compose \hat{\theta}).
\end{align*}
We proceed similarly to express path conditions of the nodes along
$\delta$.
\begin{align*}
u_1.\varphi &\equiv u.\varphi \wedge u.\theta\ese{\varphi} \equiv
u.\varphi \wedge (u.\theta \compose \theta^0)\ese{\varphi} \equiv
u.\varphi \wedge u.\theta\ese{\theta^0\ese{\varphi}}
\\
u_2.\varphi &\equiv u_1.\varphi \wedge u_1.\theta\ese{\varphi} \equiv
u.\varphi \wedge u.\theta\ese{\theta^0\ese{\varphi}} \wedge (u.\theta
\compose \theta^1)\ese{\varphi} \\ &\equiv u.\varphi \wedge
u.\theta\ese{\theta^0\ese{\varphi} \wedge \theta^1\ese{\varphi}}
\\
\cdots
\\
u_{\nu}.\varphi &\equiv u_{\nu-1}.\varphi \wedge
u_{\nu-1}.\theta\ese{\varphi} \equiv u.\varphi \wedge
u.\theta\ese{\underset{\nu}{\underbrace{\theta^0\ese{\varphi} \wedge
\ldots \wedge \theta^{\nu-1}\ese{\varphi}}}}.
\end{align*}
Using the following equivalence
\[\theta^0\ese{\varphi} \wedge \ldots \wedge
\theta^{\nu-1}\ese{\varphi} \equiv 0 \leq \nu \wedge \forall \tau~(0
\leq \tau < \nu \rightarrow \theta^{\tau}\ese{\varphi}),\]
we can write
\begin{align*}
w.\varphi &\equiv u_{\nu}.\varphi \wedge
u_{\nu}.\theta\ese{\hat{\varphi}} \\ &\equiv u.\varphi \wedge
u.\theta\ese{\theta^0\ese{\varphi} \wedge \ldots \wedge
\theta^{\nu-1}\ese{\varphi} \wedge \theta^{\nu}\ese{\hat{\varphi}}}
\\
&\equiv u.\varphi \wedge u.\theta\ese{0 \leq \nu \wedge \forall
\tau~(0 \leq \tau < \nu \rightarrow \theta^{\tau}\ese{\varphi}) \wedge
\theta^{\nu}\ese{\hat{\varphi}}}.
\end{align*}
But SMT solvers do not support the memory composition operation appearing in
the formula $w.\varphi$. Therefore, we need an equivalent \emph{declarative}
description of the operation. Such a description is a parametrised symbolic
memory $\theta_*\prm{\kappa}$, for which we require $\theta_*\prm{\kappa}
\equiv \theta^{\kappa}$, for each $\kappa \geq 0$. For a given symbolic
memory $\theta$ we compute a content of $\theta_*\prm{\kappa}$ per variable
by applying the following rules, in which \var{a} is an integer variable,
%\var{A} is an array, \
\var{b} is any variable, $c$ is a numeric constant, and $g$ is a symbolic
expression
\begin{gather*}
\frac{\theta(\var{a}) = \theta_I(\var{a}) + c%,~~~c \textrm{~is a numeric constant of \var{a}'s type}
}{\theta_*\prm{\kappa}(\var{a}) =
\theta_I(\var{a}) + c \cdot \kappa},
~~~~~~~~~~~~~
\frac{\theta(\var{a}) = \theta_I(\var{a}) \cdot c%,~~~c \textrm{~is a numeric constant of \var{a}'s type}
}{\theta_*\prm{\kappa}(\var{a}) =
\theta_I(\var{a}) \cdot c^{\kappa}},
\\ \\
\frac{\theta(\var{a}) = g,~~~\forall \theta_I(\var{b}) \in
g~.~\theta_*\prm{\kappa}(\var{b})\neq\bot}{\theta_*\prm{\kappa}(\var{a}) =
\ite(\kappa>0,\theta_*[\kappa-1]\langle g\rangle,\theta_I(\var{a}))},
% \\ \\
% \frac{\theta(\var{A}) = \theta_I(\var{A})%,~~~\var{A}~\textrm{is of an array type}
% }{\theta_*\prm{\kappa}(\var{A}) = \theta_I(\var{A})},
\end{gather*}
%where the expression $\texttt{typeOf<a>(}\kappa\texttt{)}$ represent a
%casting operation of $\kappa$ to a type of the variable \var{a}.
Observe, that lines
\ref{alg:template:repeat:foreach:BranchAritSeq}--\ref{alg:template:repeat:foreach:BranchLast}
of Algorithm~\ref{alg:template} are nothing but implementation of the rules
above. And the implementation is placed into the \textbf{repeat-until} loop
to allow application of the rules in the right order, i.e. to maximise a
chance to express \emph{all} the variables precisely.

Having $\theta_*\prm{\kappa}$ we express the resulting program state
$(x, \theta_x\prm{\kappa}, \varphi_x\prm{\kappa})$ at the location $x$ as
\begin{align*}
  \theta_x\prm{\kappa} &= \theta_*\prm{\kappa} \compose \hat{\theta}
  \\
  \varphi_x\prm{\kappa} &= 0 \leq \kappa \wedge \forall \tau~(0 \leq \tau <
  \kappa \rightarrow \theta_*\prm{\tau}\ese{\varphi}) \wedge
  \theta_*\prm{\kappa}\ese{\hat{\varphi}}
%\\
%\Xi_x\prm{\kappa} &= \theta_*\prm{\kappa} \compose \hat{\Xi},
\end{align*}
and we get $w.\theta \equiv u.\theta \compose \theta_x\prm{\nu}$, 
$w.\varphi \equiv u.\varphi \wedge u.\theta\ese{\varphi_x\prm{\nu}}$.
%and $w.\Xi \equiv u.\Xi \compose (u.\theta \compose \Xi_x\prm{\nu})$. 
Observe, that the sub-formula
\[
 0 \leq \kappa \wedge \forall \tau~(0 \leq
\tau < \kappa \rightarrow \theta_*\prm{\tau}\ese{\varphi})
\]
of $\varphi_x\prm{\kappa}$ is denoted as $\varphi_*\prm{\kappa}$ in the
algorithm (see line \ref{alg:template:endMemory}). Using the above
equivalences, we write $w.s \equiv u.s \compose (x, \theta_x\prm{\nu},
\varphi_x\prm{\nu})$, which is exactly the equivalence of~(L1).

Let $u.\varphi \wedge u.\theta\ese{\varphi_x\prm{\nu}}$ be satisfiable
formula for an exit $x$ from the cycle and for a number $\nu$ of iterations
along the core $\pi$. To prove (L2) it is sufficient to show that the path
$\delta$ (defined above) is real in $P$ and therefore it appears in $T$. For
that purpose we try to compute a path condition of \emph{classic} symbolic
execution for any path in $T$ containing $\delta$ as its suffix:
\begin{align*}
&u.\varphi \wedge u.\theta\ese{\varphi_x\prm{\nu}}~\equiv\\
&u.\varphi \wedge u.\theta\ese{0 \leq \nu \wedge \forall
\tau~(0 \leq \tau < \nu \rightarrow \theta_*\prm{\tau}\ese{\varphi}) \wedge
\theta_*\prm{\nu}\ese{\hat{\varphi}}}~\equiv\\
&u.\varphi \wedge u.\theta\ese{0 \leq \nu \wedge \forall
\tau~(0 \leq \tau < \nu \rightarrow \theta^{\tau}\ese{\varphi}) \wedge
\theta^{\nu}\ese{\hat{\varphi}}}~\equiv\\
&u.\varphi \wedge
u.\theta\ese{\theta^0\ese{\varphi} \wedge \ldots \wedge
\theta^{\nu-1}\ese{\varphi} \wedge \theta^{\nu}\ese{\hat{\varphi}}}.
\end{align*}
%where $u.\varphi$ is a classic path condition for a path from the root of $T$ down to $u$ and the formula $u.\theta\ese{\theta^0\ese{\varphi} \wedge \ldots \wedge \theta^{\nu-1}\ese{\varphi} \wedge \theta^{\nu}\ese{\hat{\varphi}}}$ uniquely identifies classic symbolic execution (see Definition~\ref{def:Composition} and Lemma~\ref{lem:EquivalentCompositions}) of the core $\pi$ exactly $\nu$ times followed by from the node $u$ is the  classic 
\qed
\end{proof}

In the following two theorems we assume that $T$ and $T'$ are classic
and compact symbolic execution trees of a given program $P$ computed by
Algorithm~\ref{alg:executeSymbolically} without and with $\Box$-lines
respectively. We further assume that neither $T$ nor $T'$ contains
failed leaves.

\begin{thm}[Soundness]
For each leaf node $e \in T$ there is a leaf node $e' \in T'$ and
a valuation  $\bld{\nu}$ of parameters in $e'.s$ such that $e.s \equiv
e'.s\prm{\bld{\nu}}$.
\end{thm}
\begin{proof}
Let $\pi$ be the path in $T$ from the root to the leaf node $e$. We
prove the theorem by the following induction:

\emph{Basic case}: The root nodes $r$ and $r'$ of $T$ and $T'$
respectively are labelled by the same program state $s_0$ (see lines
\ref{alg:executeSymbolically:InitialState} and
\ref{alg:executeSymbolically:Root} of
Algorithm~\ref{alg:executeSymbolically}). Therefore, $r.s \equiv
r'.s\prm{\bld{\nu}}$, for $\bld{\nu} = \emptyset$.

\emph{Inductive step}: Let $u \in \pi$, $u \neq e$, $u'$ be a node
of $T'$, and $\bld{\nu}$ be a valuation such that $u.s \equiv
u'.s\prm{\bld{\nu}}$. We show, there is a successor $w$ of $u$ in
$\pi$, a successor node $w'$ of $u'$ in $T'$, and a valuation
$\bld{\nu}'$ such that $w.s \equiv w'.s\prm{\bld{\nu}'}$. There are
two possible cases in Algorithm~\ref{alg:executeSymbolically} for
$u'.s$:

(1) We reach line~\ref{alg:executeSymbolically:Unrolling}: According to
Theorem~\ref{thm:l1l2}~(L1), there is a successor node $w$ of $u$ in $\pi$,
a triple $(l_i,\theta_i\prm{\kappa},\varphi_i\prm{\kappa})$ of the second
element of the applied template $t$, and a non-negative integer $\nu$ for
$\kappa$ such that
\begin{align*}
w.s &\equiv u.s \compose (l_i, \theta_i\prm{\kappa},\varphi_i\prm{\kappa})\prm{\{(\kappa,\nu)\}}
\\
&\equiv u'.s\prm{\bld{\nu}} 
  \compose(l_i,\theta_i\prm{\kappa},\varphi_i\prm{\kappa})\prm{\{(\kappa,\nu)\}}
\\
&\equiv (u'.s \compose 
  (l_i, \theta_i\prm{\kappa},\varphi_i\prm{\kappa}))\prm{\bld{\nu} \cup \{(\kappa,\nu)\} }
\\
&\equiv s'\prm{\bld{\nu}'},
\end{align*}
where $\{(\kappa,\nu)\}$ denotes a valuation assigning to $\kappa$ the
non-negative integer $\nu$, and $s'$ is the $i$-th direct successor of
$u'.s$ computed at
line~\ref{alg:executeSymbolically:SuccessorInUnrolling}. And since $w \in
T$, we have $w.\varphi$ is satisfiable. Therefore, there is a direct
successor $w'$ of $u'$ in $T'$ with $w'.s = s'$.

(2) Otherwise, we reach
line~\ref{alg:executeSymbolically:ClassicStep}: Since $u.s \equiv
u'.s\prm{\bld{\nu}}$ and we apply classic symbolic execution step for
$u'.s$, there must be a direct successor $w$ of $u$ and a direct
successor $w'$ of $u'$ such that $w.s \equiv w'.s\prm{\bld{\nu'}}$,
where $\bld{\nu'} = \bld{\nu}$.
\qed
\end{proof}

\begin{thm}[Completeness]
For each leaf node $e' \in T'$ there is a leaf node $e \in T$ and
a valuation $\bld{\nu}$ of parameters in $e'.s$ such that $e.s \equiv
e'.s\prm{\bld{\nu}}$.
\end{thm}
\begin{proof}
Let $\pi'$ be the path in $T'$ from the root to the leaf node $e'$.
We prove the theorem by the following induction:

\emph{Basic case}: The root nodes $r$ and $r'$ of $T$ and $T'$
respectively are labelled by the same program state $s_0$ (see lines
\ref{alg:executeSymbolically:InitialState} and
\ref{alg:executeSymbolically:Root}). Let us construct a non-empty set
$U$ of nodes of $T$ such that for each valuation $\bld{\nu}$ of
parameters in $r'.s$ such that $r'.\varphi\prm{\bld{\nu}}$ is
satisfiable, there is $u \in U$ such that $u.s \equiv
r'.s\prm{\bld{\nu}}$. Obviously $U = \{ r \}$, because $r'.\varphi$
contains no parameter (so $r.s \equiv r'.s\prm{\bld{\nu}}$, for each
$\bld{\nu})$.

\emph{Inductive step}:  Let $u' \in \pi'$, $u' \neq e'$ and $U$ be a
non-empty set of nodes of $T$ such that for each valuation
$\bld{\nu}$ of parameters in $u'.s$ such that
$u'.\varphi\prm{\bld{\nu}}$ is satisfiable, there is $u \in U$ such
that $u.s \equiv u'.s\prm{\bld{\nu}}$. We show, there is a successor
$w'$ of $u'$ in $\pi'$ and a non-empty set $W$ of nodes of $T$ such
that for each valuation $\bld{\nu}'$ of parameters in $w'.s$ such that
$w'.\varphi\prm{\bld{\nu}'}$ is satisfiable, there is $w \in W$ such
that $w.s \equiv w'.s\prm{\bld{\nu}'}$. And we further show that each
$w \in W$ is a successor of some $u \in U$. There are two possible
cases in Algorithm~\ref{alg:executeSymbolically} for $u'.s$:

(1) We reach line~\ref{alg:executeSymbolically:Unrolling}: Let $w'$ be a
direct successor of $u'$ in $\pi'$. Obviously, $w'.s$ is one of the states
$s'$ computed at
line~\ref{alg:executeSymbolically:SuccessorInUnrolling}. Let $i$ be the
index, for which $w'.s = s'$. The formula $w'.\varphi$ is satisfiable, since
$w'$ is in $T'$ (see condition at
line~\ref{alg:executeSymbolically:AddSATSuccessors}). Let $\bld{\nu}$ be a
valuation for which $w'.\varphi$ is satisfiable. And let $\bld{\nu}' =
\bld{\nu} \smallsetminus \{ (\kappa,\nu) \}$, where $\nu$ is an integer
assigned in $\bld{\nu}$ to the fresh parameter $\kappa$ introduced at
line~\ref{alg:executeSymbolically:UnrollingFreshKappa}. From
line~\ref{alg:executeSymbolically:SuccessorInUnrolling} we see that
$u'.\varphi\prm{\bld{\nu}'}$ is satisfiable. Therefore, there is a node $u
\in U$ such that $u.s \equiv u'.s\prm{\bld{\nu}'}$.  According to
Theorem~\ref{thm:l1l2}~(L2) there is a successor $w$ of $u$ in $T$ such that
\begin{align*}
w.s &\equiv u.s \compose (l_i, \theta_i\prm{\kappa},
\varphi_i\prm{\kappa})\prm{\{(\kappa,\nu)\}}
\\
&\equiv u'.s\prm{\bld{\nu}'} \compose (l_i, \theta_i\prm{\kappa},
\varphi_i\prm{\kappa})\prm{\{(\kappa,\nu)\}}
\\
&\equiv (u'.s \compose (l_i, \theta_i\prm{\kappa},
\varphi_i\prm{\kappa}))\prm{\bld{\nu}}
\\
&\equiv w'.s\prm{\bld{\nu}}.
\end{align*}
Therefore, $w \in W$.

(2) Otherwise, we reach
line~\ref{alg:executeSymbolically:ClassicStep}: Let $u$ be any node
in $U$. Since $u.s \equiv u'.s\prm{\bld{\nu}}$ for some valuation
$\bld{\nu}$ for which $u'.\varphi\prm{\bld{\nu}}$ is satisfiable, and
since all direct successors of both $u$ and $u'$ are computed by
classic symbolic execution step, there must be a direct successor $w$
of $u$ in $T$ and a direct successor $w'$ of $u'$ in $T'$ such that
$w.s \equiv w'.s\prm{\bld{\nu'}}$, where $\bld{\nu'} = \bld{\nu}$.
% Note that both $u'.s$ and $w'.s$ have exactly the same parameters.
Therefore, $w \in W$.
\qed
\end{proof}

%}}}
%{{{ Many Cycles in Experimental Results

\section{Many Cycles in Experimental Results}
\label{sec:ManyDetected}

Table~\ref{tab:Expriments} shows surprisingly high numbers of cycles
detected in benchmarks and relatively low numbers of computed templates.
This discrepancy can be easily explained. 

Our experimental tool first translates a given source code into a LLVM byte
code and the byte code is then translated into a flowgraph.  LLVM has an
instruction \texttt{icmp} to evaluate equality or inequality predicates.
For example, the line of LLVM code depicted in Figure~\ref{fig:cse:icmp_rep}
assigns the result of the comparison \texttt{a != 0} to~\texttt{c}. This
instruction is translated into a flowgraph depicted also in
Figure~\ref{fig:cse:icmp_rep}.

\begin{figure}[!htb]
  \begin{center}
    \begin{tabular}{ccc}
      \begin{tabular}{c}
        \texttt{\%c = icmp ne i32 \%a, 0}
      \end{tabular}
      & ~~~~~~~~~~~~~~ &
      \begin{tabular}{c}
%        \includegraphics[width=2.5cm]{fig_cse_icmp_rep}
%        \centering
        \tikzstyle{loc} = [circle,thick,draw,minimum size=4mm]
        \tikzstyle{pre} = [<-,shorten <=1pt,>=stealth',semithick]
        \footnotesize
        \begin{tikzpicture}[node distance=1.0cm]    
          \node [loc] (1) {};
          \node [loc] (2) [below of=1,xshift=-12mm] {}
            edge [pre] node [label=left:\texttt{a!=0}\,\,] {} (1);
          \node [loc] (3) [below of=1,xshift=12mm] {}
            edge [pre] node [label=right:~\texttt{a=0}] {} (1);
          \node [loc] (4) [below of=2,xshift=12mm] {}
            edge [pre] node [label=left:\texttt{c:=1}~~] {} (2)
            edge [pre] node [label=right:\texttt{c:=0}~~\,\,] {} (3);
        \end{tikzpicture}
      \end{tabular}
    \end{tabular}
  \end{center}
  \caption{A LLVM instruction \texttt{icmp} and its flowgraph representation.}
  \label{fig:cse:icmp_rep}
\end{figure}
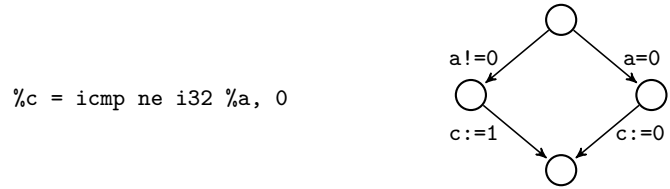

As shown in Appendix~\ref{sec:LoopsCycles}, branching structures inside
program loops lead to a high number of cycles in flowgraphs. Hence, if there
are \texttt{icmp} instructions in loops of an LLVM byte code, then we detect
many more cycles in the resulting flowgraph compared to the number of loops
in the LLVM byte code. More precisely, the number of cycles can grow
exponentially in the number of \texttt{icmp} instructions inside a program
loop. \todo{Je to pravda, ze jich je exponencialne? V zakladnim pripade ano,
  ale lze to rict takto obecne?}

However, only a few of the detected cycles are feasible in practice.
\marek{Chceme tu rozpatlavat duvody? Ve zdrojaku je jedna zakomentovana
  veta, ale ta toho vic nakousne nez dorekne. Chtelo by to pak dalsi
  obrazek. Marku, ty bys to mohl v disertacce vysvetlit i s dalsim obrazkem,
  ne?}
% (\texttt{icmp} instructions typically correspond to (in)equality checks and
% they are often followed by conditional jumps depending on the \texttt{icmp}
% results).
As templates are computed only for (provably) feasible cycles (see
line~\ref{alg:template:satCheck1} of Algorithm~\ref{alg:template}), we
usually get a relatively low number of templates.
\todo{Zkontrolujte prosim, jestli je ten posledni odstavec dobre.}

%}}}

\end{document}